\documentclass[peerreview, 10pt]{IEEEtran}
\usepackage{graphics}
\usepackage{epsfig}

\usepackage{subfigure}
\usepackage{graphics}

\usepackage{amsmath}
\usepackage{amsfonts}
\usepackage{amssymb}
\usepackage{stfloats}
\usepackage{tikz}
\usetikzlibrary{arrows}
\usetikzlibrary{arrows,decorations.pathmorphing,backgrounds,fit}
\tikzstyle{block}=[draw opacity=0.7,line width=1.4cm]

\newcommand{\rem}[1]{}

\newcommand{\onec}
{
\left(
\begin{array}{c}
1 \\
\bc
\end{array}
\right)
}

\newtheorem{proposition}{Proposition}
\newtheorem{lemma}{Lemma}

\newtheorem{algorithm}{Algorithm}

\newcommand{\bre}{\begin{equation}}
\newcommand{\ere}{\end{equation}}

\newcommand{\ee}\]
\newcommand{\bra}{\begin{eqnarray}}
\newcommand{\era}{\end{eqnarray}}
\newcommand{\bfg}{\begin{figure}[hbtp]}
\newcommand{\efg}{\end{figure}}
\newcommand{\bver}{\begin{verbatim}}
\newcommand{\ever}{\end{verbatim}}
\newcommand{\bit}{\begin{itemize}}
\newcommand{\eit}{\end{itemize}}
\newcommand{\ben}{\begin{enumerate}}
\newcommand{\een}{\end{enumerate}}

\newcommand{\bmu}{\mbox{\boldmath $\mu$}}

\newcommand{\given}{\: | \:}

\newcommand{\bomega}{\mbox{\boldmath $\omega$}}

\newcommand{\btlambda}{\tilde{\mbox{\boldmath $\lambda$}}}

\newcommand{\bzero}{{\bf 0}}

\newcommand{\bg}{{\bf{g}}}

\newcommand{\bc}{{\bf c}}

\newcommand{\bhc}{\hat {\bf c}}

\newcommand{\tlambda}{\tilde {\lambda}}

\newcommand{\tPsi}{\tilde{\Psi}}

\newcommand{\baa}{\begin{eqnarray*}}
\newcommand{\eaa}{\end{eqnarray*}}

\newcommand{\bs}{{\bf s}}

\newcommand{\bb}{{\bf b}}

\newcommand{\bh}{{\bf h}}
\newcommand{\bu}{{\bf u}}
\newcommand{\bv}{{\bf v}}

\newcommand{\bx}{{\bf x}}

\newcommand{\bone}{{\bf 1}}
\newcommand{\bzr}{{\bf 0}}

\newcommand{\cA}{{\cal A}}
\newcommand{\cP}{{\cal P}}
\newcommand{\cQ}{{\cal Q}}
\newcommand{\cE}{{\cal E}}
\newcommand{\cF}{{\cal F}}
\newcommand{\cR}{{\cal R}}

\newcommand{\cT}{{\cal T}}

\newcommand{\cH}{{\cal H}}

\newcommand{\cV}{{\cal V}}
\newcommand{\cJ}{{\cal J}}

\newcommand{\cC}{{\mathcal{C}}}
\newcommand{\cI}{{\mathcal{I}}}

\newcommand{\cN}{{\mathcal{N}}}
\newcommand{\cG}{\mathcal{G}}

\newcommand{\btc}{\tilde{\bf c}}

\newcommand{\hbc}{\hat{\bf c}}
\newcommand{\hbg}{\hat{\bf g}}

\newcommand{\beginproof}{\noindent \textbf{Proof: }  }
\newcommand{\finproof}{\noindent $\Box$\\}

\def\defined{\: {\stackrel{\scriptscriptstyle \Delta}{=}} \: }

\def\argmin{\mathop{\rm argmin}}

\def\defined{\: {\stackrel{\scriptscriptstyle \Delta}{=}} \: }

\newfont{\boldlarge}{msbm10 scaled 1100}

\newcommand{\D}{{\rm D}}
\renewcommand{\P}{{\rm P}}
\newcommand{\DS}{{\rm DS}}

\newcommand{\comment}[1]{}

\def\etal{$et \,\,al.\,\,$}

\renewcommand{\thesection}{\Roman{section}}

\newlength{\tmpbigbar}

\rem{
\newcommand{\D}{{\rm D}}
\renewcommand{\P}{{\rm P}}
\newcommand{\DS}{{\rm DS}}

}

\normalsize
\begin{document}

\title{
Improved linear programming decoding of LDPC codes and bounds on the
minimum and fractional distance}

\author{\begin{tabular}{c c} David~Burshtein,~\IEEEmembership{Senior Member,~IEEE} & Idan~Goldenberg,~\IEEEmembership{Student Member,~IEEE}\\
School of Electrical Engineering & School of Electrical Engineering \\ Tel-Aviv University, Tel-Aviv 69978, Israel & Tel-Aviv University, Tel-Aviv 69978, Israel\\
Email:  burstyn@eng.tau.ac.il  & Email:  idang@eng.tau.ac.il
\end{tabular} }

\markboth{Submitted to IEEE Transactions on Information Theory}
{Burshtein and Goldenberg: Iterative Approximate Linear Programming
Decoding of Linear Codes with Linear Complexity}

\maketitle 
\thanks{This research was supported by the Israel Science Foundation, grant no. 772/09.}
\IEEEpeerreviewmaketitle

\begin{abstract}
We examine LDPC codes decoded using linear programming (LP). Four
contributions to the LP framework are presented. First, a new method
of tightening the LP relaxation, and thus improving the LP decoder,
is proposed. Second, we present an algorithm which calculates a
lower bound on the minimum distance of a specific code. This
algorithm exhibits complexity which scales quadratically with the
block length. Third, we propose a method to obtain a tight lower
bound on the fractional distance, also with quadratic complexity,
and thus less than previously-existing methods. Finally, we show how
the fundamental LP polytope for generalized LDPC codes and nonbinary
LDPC codes can be obtained.
\end{abstract}


\section{Introduction} \label{sec:Introduction}
Within the field of error-correcting codes, the method of linear
programming (LP) decoding of linear codes has attracted considerable
attention in recent years. Over the past several years, this method
has stirred a substantial amount of research. One reason for this is
that the LP decoder has the ML certificate property~\cite{lpdecode},
i.e., that if the LP decoder succeeds in producing an integral
solution vector, then it produces the ML solution. Another
attractive property of the LP decoder is that one may improve the
performance in various ways. These include using generic techniques
to tighten the LP relaxations~\cite{lpdecode}, using integer
programming or mixed integer linear
programming~\cite{lpdecode},~\cite{draper2007mdv}, adding
constraints to the Tanner graph~\cite{taghavi2008amlp} and guessing
facets of the polytope~\cite{dimakis2007gfp}.

However, LP decoding is disadvantaged as compared to iterative
decoding algorithms such as belief propagation or min-sum decoding
because the decoding complexity can be much larger. In general the
LP decoder entails a polynomial, non-linear complexity. Several
authors have proposed algorithms for low-complexity LP decoding,
e.g.~\cite{low_complexity_LP},~\cite{low_complexity_LP_journal},~\cite{yang2008nlp}.
Vontobel and Koetter~\cite{low_complexity_LP} have proposed an
iterative, Gauss-Seidel-type algorithm for approximate LP decoding.
In~\cite{lpon_journal} the approach of~\cite{low_complexity_LP} was
studied under a newly-proposed scheduling scheme. This study showed
that the performance of the LP decoder can be approached with linear
computational complexity. It was also shown that the iterative LP
decoder can correct a fixed fraction of errors with linear
computational complexity. These results were also extended to
generalized LDPC (GLDPC) codes.

In this paper, we present the following main results. First, we show
that merging check nodes in the graph can tighten the relaxation and
improve the LP decoder. Second, we demonstrate how the
linear-complexity LP decoder from \cite{lpon_journal} can be used to
obtain a lower bound on the minimum distance of a specific code,
which is also an upper bound on the fractional distance of the code,
with complexity $O(N^2)$, where $N$ is the block length. This lower
bound on the minimum distance can be improved by using the check
node merging technique. Third, we propose an algorithm which
calculates a tight lower bound on the fractional distance with
complexity $O(N^2)$. We note that Feldman \etal \cite{lpdecode} were
the first to propose an algorithm which calculates the exact
fractional distance; their algorithm entails a complexity which is
polynomial but not quadratic in the block length. This is also true
for the fractional distance calculation in \cite{yang2008nlp}, even
though it has reduced complexity as compared to \cite{lpdecode}. Our
approach has lower complexity as compared to these techniques.
Finally, we show how the fundamental polytope of two important
classes of codes, namely GLDPC and nonbinary codes, can be obtained.
This is accomplished by using the double description method
\cite{FukudaProdon1995,Motzkin}. In addition to these results, we
show how the linear-complexity LP decoder can be calculated
efficiently for GLDPC codes.

This paper is organized as follows. In section~\ref{sec:linprogdec}
the linear-complexity LP decoder~\cite{lpon_journal} is reviewed.
Section~\ref{sec:eff_comp} shows how to perform some computations
relating to the linear-complexity LP decoder in an efficient manner.
Section~\ref{sec:cuts for_improved_LPD} describes how improved LP
decoding can be achieved by merging check nodes. In
Section~\ref{sec:min distance bound} we propose an algorithm which
yields a lower bound on the minimum distance which is also an upper
bound on the fractional distance of specific LDPC codes. In
Section~\ref{sec: lower bound on dfrac}, an algorithm which produces
a tight lower bound on the fractional distance is presented. We show
how to obtain the fundamental polytope of GLDPC and nonbinary codes
in Section~\ref{sec:GLDPC_polytope}. Section~\ref{sec: results}
provides some numerical results, and the paper is concluded in
Section~\ref{sec: conclusion}.

\section{Approximate Iterative Linear programming decoding of GLDPC codes} \label{sec:linprogdec}

Consider a discrete binary-input memoryless channel (DMC) described
by the probability transition function $Q(y \given c)$ where $c$ is
the transmitted code-bit and $y$ is the channel output. Also
consider a GLDPC code $\cC$ represented by a Tanner graph with $N$
variable nodes and $M$ linear constraint nodes, where each
constraint node represents a small local constituent code, thus
generalizing single parity checks (SPCs) in plain LDPC codes. The
code $\cC$ is used to transmit information over the given DMC.
Following the notation in~\cite{lpdecode}, let $\cI$ and $\cJ$ be
the sets of variable and constraint nodes, respectively, such that
$|\cI| = N$ and $|\cJ|=M$. The variable node $i\in\cI$ is connected
to the set $\cN_i$ of constraint node neighbors. The constraint node
$j\in\cJ$ is connected to the set $\cN_j$ of constraint node
neighbors. Further denote by $\bc_j$ the vector \bre
\label{eq:cjvec_def} \bc_j
\defined \{ c_{i} \}_{i\in \cN_j} \: . \ere Given some
$j\in\cJ$, denote by $\cC_j$ the constituent binary linear code
corresponding to the constraint node $j\in\cJ$ (thus for plain LDPC
codes, $\cC_j$ is a simple parity-check code). Let $\bzero_j$ denote
the all-zero local codeword on constraint node $j$, and let $\bzero$
denote the all-zero codeword in $\cC$. A valid codeword $\bc \in
\cC$ satisfies $\bc_j \in \cC_j$, $\forall j\in\cJ$.

We assume, without loss of generality, that there are no parallel
edges in the Tanner graph, since otherwise we can modify the
structure of the code to satisfy this constraint.

The LP decoder is given by~\cite{lpdecode},

\bre {\hbc} \defined \argmin_{\bc,\bomega} \P(\bc) \label{eq: lp
problem} \ere where the vector $\bomega$ is defined by \bre
\label{eq:w_def} \bomega
\defined \{ w_{j,\bg} \}_{j\in\cJ,\: \bg \in \cC_j} \nonumber \ere
and where \bre \P(\bc) \defined \sum_{i\in\cI} c_i \gamma_i
\label{eq:mllp} \ere subject to \bre w_{j,\bg} \ge 0 \qquad \forall
j \in \cJ \: , \: \bg \in \cC_j \label{eq:positive_w} \ere \bre
\sum_{\bg \in \cC_j} w_{j,\bg} = 1 \qquad \forall j \in \cJ
\label{eq:w_sum_1} \ere \bre c_i = \sum_{\bg\in \cC_j \: , \: g_i=1}
w_{j,\bg} \qquad \forall j\in\cJ \: , \: i \in \cN_j \label{eq:ci}
\ere and where \bre \gamma_i
\defined \log \frac{Q(y_i \given 0)}{Q(y_i \given 1)}
\label{eq:gamma_def} \nonumber \ere where all logarithms henceforth
are to the base $e$.  We assume that there are no perfect
measurements from the channel, i.e. \bre -\infty < \gamma_i < \infty
\qquad \forall i\in\cI \: . \label{eq:lossy_channel} \ere Note that
if we transmit over the binary erasure channel (BEC) then this
assumption does not hold. However, we can re-normalize the
$\gamma_i$'s by replacing $\gamma_i=\infty$ with $\gamma_i=1$ and
$\gamma_i=-\infty$ with $\gamma_i=-1$. After this re-normalization
the assumption~\eqref{eq:lossy_channel} holds, and we can then apply
the decoding algorithm.

We term the minimization problem~\eqref{eq: lp
problem}-\eqref{eq:ci}, {\em Problem-P}. We also define \bre
\label{eq:wjsvec_def} \bomega_j
\defined \{ w_{j,\bg} \}_{\bg\in \cC_j} \nonumber \ere
The polytope $\cQ_j(\cC)$ is defined as follows~\cite{lpdecode}:
$\bc\in\cQ_j(\cC)$ (sometimes we will use the notation
$\bc_j\in\cQ_j(\cC)$, the meaning will be clear from the context) if
and only if there exists a vector $\bomega_j$ such
that~\eqref{eq:positive_w}-\eqref{eq:ci} hold for this value of $j$.
For plain LDPC codes it was shown that $\bc\in\cQ_j(\cC)$ if and
only if \bre \label{eq:omegaj_def2}
\begin{array}{l l} 0 \le c_i \le 1 &  \forall i\in\cN_j \\
\sum_{i\in\cN_j\setminus S} c_i + \sum_{i\in S} \left( 1-c_i \right)
\ge 1 & \forall S \subseteq \cN_j \: , |S| \: {\rm odd} \end{array}
\ere We further denote by $\cQ(\cC) \defined \bigcap_{j\in\cJ}
\cQ_j(\cC)$ the relaxed~\cite{lpdecode} or fundamental polytope.
Thus $\bc$ is feasible in~\eqref{eq:positive_w}-\eqref{eq:ci} if and
only if $\bc\in\cQ(\cC)$. We thus have, in addition to \eqref{eq: lp
problem}, the relation \bre \hbc = \argmin_{\bc \in \cQ(\cC)}
\P(\bc) \label{eq: lp problem, no w} \ere

An important observation is that the LP decoder has the ML
certificate~\cite{lpdecode} in the sense that if the solution of
Problem-P is integer (in fact an integer $\bomega$ implies that
$\bc$ is also integer by~\eqref{eq:ci}) then it is the ML decision.

Now consider the following problem, which we term {\em Problem-D}.
The variables in this problem are
$$
\bu = \{u_{i,j}\}_{i\in \cI,j\in \cN_i}
$$
and the problem is defined by \bre \max_{\bu} \D(\bu) \quad {\rm
where} \quad \D(\bu) \defined \sum_{j\in \cJ} \min_{\bg\in \cC_j}
\left\{ \sum_{i\in \cN_j} u_{i,j} g_i \right\} \label{eq:lp_dual}
\ere subject to \bre \sum_{j\in \cN_i} u_{i,j} = \gamma_i \qquad
\forall i\in \cI \: . \label{eq:lp_dual_constraint} \ere

\begin{lemma}
\label{lem:lp_dual} Problem-P and Problem-D are Lagrange dual
problems. The minimum of Problem-P is equal to the maximum of
Problem-D.
\end{lemma}
The lemma is proven in~\cite{lpon_journal}.

Vontobel and Koetter~\cite{low_complexity_LP} have proposed an
algorithm which yields an approximated solution to Problem-D. They
suggested to consider the following ``soft'' version of Problem-D,
termed {\em Problem-DS}: \bre \max_{\bu} \DS(\bu)
\label{eq:lp_duals} \nonumber \ere where \bre \DS(\bu)
\defined -\frac{1}{K} \sum_{j\in \cJ} \log \sum_{\bg\in \cC_j} e^{-K
\sum_{i\in \cN_j} u_{i,j} g_i} \label{eq:lp duals objf} \ere subject
to~\eqref{eq:lp_dual_constraint}. Note that the objective function
in~\eqref{eq:lp duals objf} is obtained from~\eqref{eq:lp_dual} by
approximating the $\min()$ function using a log-sum-exp function
with a parameter $K>0$. This approximation can become arbitrarily
accurate by increasing $K$. Also note that $\DS(\bu)$ is concave
since the log-sum-exp function is convex. In what follows, we give a
brief overview of the linear-complexity iterative LP decoder. The
reader is referred to \cite{lpon_journal} for additional details.

Let $\tPsi^j$ denote the $(|\cN_j|+1) \times |\cC_j|$ matrix \bre
\label{eq:tpsi_j_def} \tPsi^j \defined \left(
\begin{array}{c}
  \bone_{|\cC_j|} \\
  \Psi^j
\end{array}
\right) \ere where $\bone_{|\cC_j|}$ is a row vector of $|\cC_j|$
ones and the columns of $\Psi^j$ are all the codewords of $\cC_j$.

Note that the equation set \bre \label{eq:w1_ci_mat_form} \left(
\begin{array}{c}
  1 \\
  \bc_j
\end{array}
\right) = \tPsi^j \bomega_j \nonumber\ere expresses the
constraints~\eqref{eq:w_sum_1}-\eqref{eq:ci}. We define
$$
\Xi^j \defined \left( \tPsi^j \right)^T \left( \tPsi^j
\left(\tPsi^j\right)^T \right)^{-1}
$$
where $A^T$ is the transpose of matrix $A$, and
$$
\eta \defined \max_j \left\{ |\cC_j| |\cN_j| \max_{l,i} |
\Xi^j_{l,i} | \right\}
$$
where $\Xi^j_{l,i}$ is the $(l, i)$ element of $\Xi^j$.

Now let $\epsilon>0$ be a constant such that $\epsilon \le 1/\eta$.
We define \bre \label{eq:tlamdef_GLDPC} \tlambda_i
\defined \lambda_i (1-\eta\epsilon) + \frac{\eta}{2}\epsilon \qquad
\forall i\in\cI \ere and \bre \label{eq:btlambda_def} \btlambda
\defined \{ \tlambda_i \}_{i\in\cI} \: . \nonumber \ere
We will also define the quantities \bra \label{eq:B_kj_A_kj} A_{k,j}
&\defined& \sum_{\bg\in \cC_j, \:
g_k=1} e^{-K\sum_{i\in \cN_j\setminus k} u_{i,j} g_i} \nonumber \\
B_{k,j} & \defined & \sum_{\bg\in \cC_j, \: g_k=0} e^{-K\sum_{i\in
\cN_j\setminus k} u_{i,j} g_i} \era The following iterative
algorithm was proposed in~\cite{lpon_journal} to find an approximate
solution to Problem-P in linear time complexity.

\begin{algorithm}
\label{alg} Given a GLDPC code, channel observations and some fixed
parameters $\epsilon_0\in(0,1/\eta)$ and $K>0$ do:

\begin{enumerate}
\item
{\bf Initialization:}
\begin{align*}
\forall i\in\cI, j\in\cN_i \: : \: \: & u_{i,j} = \gamma_i / |\cN_i|
\\
\forall i\in\cI, j\in\cN_i \: : \: \: & v_{i,j} = -\frac{1}{K} \log
\frac{B_{i,j}}{A_{i,j}}
\\
\forall i\in\cI, j\in\cN_i \: : \: \: & \lambda_{i,j} = \frac{1}{1 +
e^{K(u_{i,j}-v_{i,j})}}
\\
\forall i\in\cI \: : \: \: & \lambda_{i} = \frac{1}{|\cN_i|}
\sum_{j\in\cN_i} \lambda_{i,j}
\\
\forall i\in\cI \: : \: \: & \epsilon_i = \max_{j\in\cN_i}
|\lambda_{i,j}-\lambda_i|
\\
&\cA = \{ i\in\cI \: : \: \epsilon_i \ge \epsilon_0 \}
\end{align*}
\item
{\bf Iteration:} Pick an arbitrary element $k \in \cA$ and make the
following updates:
\begin{align*}
&\forall j\in\cN_k  : \: u_{k,j} = v_{k,j} +
\frac{\gamma_k}{|\cN_k|} - \frac{1}{|\cN_k|} \sum_{l\in\cN_k}
v_{k,l}
\\
&\forall i\in\cN_j\setminus k, j\in\cN_k :  \: v_{i,j} =
-\frac{1}{K} \log \frac{B_{i,j}}{A_{i,j}}
\\
&\forall i\in\cN_j, j\in\cN_k:  \: \lambda_{i,j} = \frac{1}{1 +
e^{K(u_{i,j}-v_{i,j})}}
\\
&\forall i\in\cN_j, j\in\cN_k  :  \: \lambda_{i} = \frac{1}{|\cN_i|}
\sum_{j'\in\cN_i} \lambda_{i,j'}
\\
&\forall i\in\cN_j, j\in\cN_k :  \: \epsilon_i = \max_{j'\in\cN_i}
|\lambda_{i,j'}-\lambda_i|
\\
&\forall i\in\cN_j, j\in\cN_k :  \: \mbox{If } i\in\cA \mbox{ and }
\epsilon_i < \epsilon_0 \mbox{ then } \cA = \cA \setminus i
\\
&\forall i\in\cN_j, j\in\cN_k :  \: \mbox{If } i\not\in\cA \mbox{
and } \epsilon_i \ge \epsilon_0 \mbox{ then } \cA = \cA \cup i
\end{align*}
\item
{\bf Loop Control:} If $\cA \ne \emptyset$ then repeat step 2.
Otherwise proceed to step 4.
\item
{\bf Produce the final solution and Exit:} $\forall i\in\cI$ :
Compute $\tlambda_i$ using~\eqref{eq:tlamdef_GLDPC}, and $\hat{c}_i$
using \bre \hat{c}_i = \left\{
  \begin{array}{ll}
    1 & \hbox{if $\tlambda_i>0.5$} \\
    0 & \hbox{if $\tlambda_i<0.5$}
  \end{array}
\right. \label{eq:hci} \nonumber \ere Output the primal vector
$\btlambda$, the dual vector $\bu$ and the bitwise-estimated
integral vector $\{\hat{c}_i\}_{i \in \cI}$

\end{enumerate}
\end{algorithm}

Note that each edge in the graph connecting the variable node $i$
and the constraint node $j$ is associated with the variables
$u_{i,j}$, $v_{i,j}$ and $\lambda_{i,j}$. Each variable node
$i\in\cI$ is associated with the variables $\lambda_i$,
$\tlambda_i$, $\hat{c}_i$ and
$\epsilon_i$. 
In~\cite{lpon_journal} the following assumptions were made.
\begin{enumerate}
\item
\bre | \gamma_i | \le \gamma_{\max} \qquad \forall i\in\cI
\label{eq:gammam} \ere for some positive constant $\gamma_{\max}$.
Note that practical channels with non-perfect measurements, that
satisfy~\eqref{eq:lossy_channel}, also satisfy~\eqref{eq:gammam},
since their output is quantized to some finite set.
\item
As was indicated earlier, we can assume without loss of generality that there are no parallel edges in the Tanner graph.
\item
We assume that $d'_j \ge 3$ $\forall j\in\cJ$. where $d'_j$ is the
dual distance of the constituent code corresponding to the
constraint node $j$. For plain LDPC codes this assumption
degenerates to $|\cN_j|\ge 3$ $\forall j\in\cJ$ (this simpler
version appears also in~\cite{low_complexity_LP}).
\item
We assume that all node degrees are bounded by some constant (independent of $N$).
\end{enumerate}
Under these assumptions it was shown in \cite{lpon_journal}, for any
fixed $\epsilon_0\in(0,1/\eta)$ and $K>0$, that Algorithm~\ref{alg}
exits after $O(N)$ iterations, and that the vector $\btlambda$ is
feasible in Problem-P. Furthermore, by setting $K$ sufficiently
large and $\epsilon_0$ sufficiently small, we can make
$(P(\btlambda) - P^\ast) / N$, where $\P^\ast$ is the minimum value
of Problem-P, arbitrarily small. More precisely,

\cite[Theorem 1]{lpon_journal}: Consider Algorithm~\ref{alg} under
the four assumptions above. For any given $\delta$, after $O(N)$
iterations, each with a constant number of operations, the algorithm
yields a vector, $\btlambda$, which is feasible in Problem-P and
satisfies \bre 0 \le \P(\btlambda) - \P^{\ast} \le \delta N
\label{eq:lincomp} \ere

By weak duality and by the primal feasibility of $\btlambda$, we
have \bre \label{eq: sandwich} \D(\bu) \le \P^\ast \le \P(\btlambda)
\ere We will make use of the fact that $\P^\ast$ is sandwiched
between the lower bound $\D(\bu)$ and the upper bound
$\P(\btlambda)$. In fact, in the proof of \cite[Theorem
1]{lpon_journal}, it was shown that \bre \label{eq: duality gap}
\P(\btlambda)-\D(\bu) \le \delta N \ere and thus follows the result
of the theorem.

\section{Efficient computations related to the iterative approximate LP decoder} \label{sec:eff_comp}
In this section, we demonstrate how the computation of $A_{k,j}$ and
$B_{k,j}$, defined in~\eqref{eq:B_kj_A_kj} can be efficiently
implemented for GLDPC codes. For LDPC codes an efficient computation
scheme was presented in \cite{lpon_journal}. We also show how to
efficiently implement the calculation of $\D(\bu)$.

We follow the approach proposed in~\cite{wolf1978efficient} that
uses a trellis representation of the code. To simplify the notation
we assume that $\cN_j = \{ 0,1,\ldots,n-1 \}$, and we fix $j$ and
omit it from $A_{k,j}$, $B_{k,j}$, $\cC_j$ and $u_{i,j}$. In
addition we assume that the code $\cC_j$ is represented by the
parity check matrix
$$
H = \left[ \bh_0 \: \bh_1 \: \ldots \: \bh_{n-1} \right]
$$
with columns $\bh_0, \bh_1, \ldots, \bh_{n-1}$, such that the
vectors $\bh_i$ are $m$-dimensional binary vectors representing
elements of the field ${\rm GF}(2^m)$.

\subsection{Efficient Computation of $A_{k,j}$ and $B_{k,j}$} \label{sec: Akj and Bkj}
We need to compute $A_k$ and $B_k$, $k=0,1,\ldots,n-1$, defined by
\bra \label{eq:B_k_A_k} A_{k} &\defined& \sum_{\bg \: : \: H \bg =
0, \: g_k=1} e^{-K \left( \sum_{i=0}^{n-1} u_{i} g_i - u_k \right)}
\nonumber \\ B_{k} &\defined& \sum_{\bg \: : \: H \bg = 0, \: g_k=0}
e^{-K \sum_{i=0}^{n-1} u_{i} g_i} \: . \era where $\bg
\defined \left[ g_0, g_1, \ldots, g_{n-1} \right]^T$

Let $H_k$, $\bg_k$, $E(k,\bs)$ and $E(k,\bs,b)$ be defined by \bre
H_k \defined \left[ \bh_0 \: \bh_1 \: \ldots \: \bh_{k} \right]
\quad , \quad \bg_k \defined \left[ g_0, g_1, \ldots, g_{k}
\right]^T \label{eq: def of Hk gk}\ere
$$
E(k,\bs) = \sum_{\bg_k \: : \: H_k \bg_k = \bs} e^{-K \sum_{i=0}^k u_i g_i}
$$
and \bra E(k,\bs,b) &=& \sum_{\bg_k \: : \: H_k \bg_k = \bs, g_k =
b} e^{-K \sum_{i=0}^k u_i g_i} \nonumber \\ &=& E(k-1,\bs \oplus b
\cdot \bh_k) e^{-K u_k \cdot b} \era In these definitions $\bs$ is
an $m$ dimensional binary vector, representing an element in ${\rm
GF}(2^m)$, $b \in \{ 0, 1 \}$, $g_i \in \{ 0, 1 \}$, and $\oplus$ is
the bitwise exclusive or operator (addition in ${\rm GF}(2^m)$).

Similarly we define \bra \tilde{H}_k &\defined& \left[ \bh_{k+1} \:
\bh_{k+2} \: \ldots \: \bh_{n-1} \right] \nonumber \\ \tilde{\bg}_k
&\defined& \left[ g_{k+1}, g_{k+2}, \ldots, g_{n-1} \right]^T \era
$$
\tilde{E}(k,\bs) = \sum_{\tilde{\bg}_k \: : \: \tilde{H}_k \tilde{\bg}_k = \bs} e^{-K \sum_{i=k+1}^{n-1} u_i g_i}
$$

$E(k,\bs)$ and $\tilde{E}(k,\bs)$ can be computed recursively using
the forward recursion for $k = 0,1,\ldots,n-1$ \bre E(k,\bs) =
E(k-1,\bs) + E(k-1,\bs \oplus \bh_{k}) \cdot e^{-K u_{k}}
\label{eq:forward_rec} \ere (for a given value of $k$ we use this
recursion to compute $E(k,\bs)$ for all possible states $\bs$) and
the backward recursion for $k = n-2,\ldots,1,0$ \bre
\tilde{E}(k,\bs) = \tilde{E}(k+1,\bs) + \tilde{E}(k+1,\bs \oplus
\bh_{k+1}) \cdot e^{-K u_{k+1}} \label{eq:backward_rec} \ere and the
initial conditions \bre E(-1,\bs) = \left\{
              \begin{array}{ll}
                1, & \hbox{\bs=0} \\
                0, & \hbox{otherwise}
              \end{array}
            \right.
\ere
and
\bre \tilde{E}(n-1,\bs) =
\left\{
              \begin{array}{ll}
                1, & \hbox{\bs=0} \\
                0, & \hbox{otherwise}
              \end{array}
            \right.
\ere

Finally, $A_k$ and $B_k$ are computed using
\bre \label{eq: calcAk}
A_k = e^{K u_k} \sum_{\bs} E(k,\bs,1) \tilde{E}(k,\bs)
\ere
and
\bre \label{eq: calcBk}
B_k = \sum_{\bs} E(k,\bs,0) \tilde{E}(k,\bs)
\ere
for $k=0,\ldots,n-1$.

In practice, in order to avoid numerical problems, the forward
recursion~\eqref{eq:forward_rec} and the backward
recursion~\eqref{eq:backward_rec} are computed using logarithmic
computation.

Another possibility is to use the method proposed by Hartmann and
Rudolph~\cite{hartmann1976oss}. However, in this case one faces a
numerical problem since the method involves subtractions, and not
only summations.

\subsection{Efficient calculation of $\D(\bu)$}
\label{sec: calculation of D(u)} To calculate $\D(\bu)$ for GLDPC
codes, we use the same methodology as in the calculation of
$A_{k,j}$ and $B_{k,j}$, and we use the definitions of $H_k$ and
$\bg_k$ in \eqref{eq: def of Hk gk}. The difference is that here we
only need a forward recursion for the calculation.

By defining \bre A(k,\bs) \defined \min_{\bg_k: \; H_k \bg_k=
\bs}\sum_{i=0}^{n-1} u_i g_i \nonumber \ere it is easily seen that
our objective (the minimization in \eqref{eq:lp_dual}) is equal to
\bre \label{eq: D(u) efficient} A(n-1,0)\ere Now, $A(k,\bs)$ can be
calculated recursively for $k=0,1,\dots,n-1$ as
$$
A(k,\bs) = \min(A(k-1,\bs),A(k-1,\bs \oplus \bh_k)+u_k)
$$
with the initial conditions \bre A(-1,\bs) = \left\{
              \begin{array}{ll}
                0, & \hbox{\bs=0} \\
                \infty, & \hbox{otherwise}
              \end{array}
            \right. \label{eq: D(u) initial conditions}
\ere

\section{A new method for improved LP decoding} \label{sec:cuts for_improved_LPD}
Consider some GLDPC code $\cC$ with an underlying Tanner graph
$\cG$. Let $\hbc$ be the LP decoded word considered in
Section~\ref{sec:linprogdec}, i.e.
$$
\hbc = \argmin_{\bc \in \cQ(\cC)} \P(\bc)
$$
where
\begin{multline*} \cQ(\cC) = \\ \left\{ \bc \: : \: \exists
\bomega = \left\{ w_{j,S} \right\}_{j\in\cJ,\: S\in\cE_j} \mbox{
s.t. \eqref{eq:positive_w}, \eqref{eq:w_sum_1} and \eqref{eq:ci}
hold} \right\}
\end{multline*}

Now suppose that we merge some of the constraint nodes together to
form an alternative representation of the given code, with
underlying graph $\cG'$. This is illustrated in
Figure~\ref{fig:merge_example}, where we merge the constraint nodes
$j_1$ and $j_2$ into $j'_1$ (i.e. the constituent code associated
with $j'_1$ satisfies both the constraints implied by $j_1$ and
those implied by $j_2$), and the constraint nodes $j_3$ and $j_4$
into $j'_2$. In this example the constraint node $j_5$ is left
unchanged ($j'_3$ in $\cG'$).
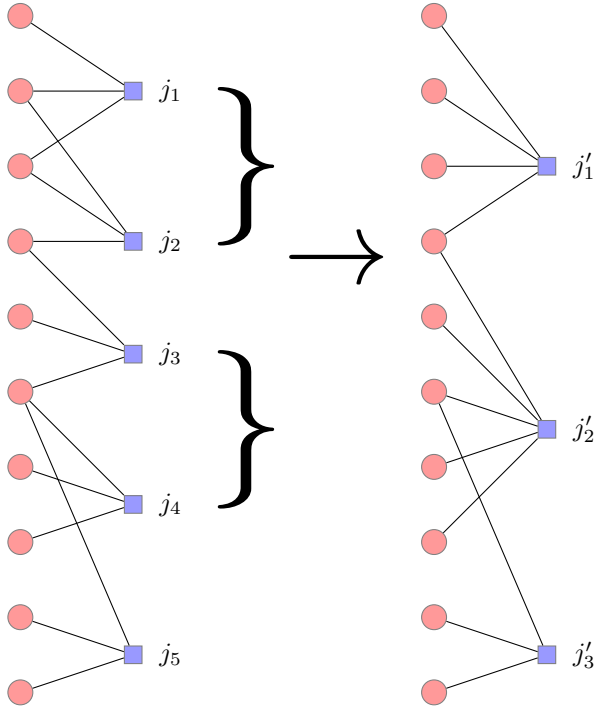
\begin{figure}[hbtp]
\begin{center}
\begin{tikzpicture}
[var/.style={circle,draw=black!50,fill=red!40},
par/.style={rectangle,draw=black!50,fill=blue!40},
txt/.style={},
txt2/.style={scale=2},
txt3/.style={scale=3},
txt4/.style={scale=4},
txt5/.style={scale=5},
txt6/.style={scale=6},
]


\node (v1)  at (0,10)  [var] {};
\node (v2)  at (0,9)   [var] {};
\node (v3)  at (0,8)   [var] {};
\node (v4)  at (0,7)   [var] {};
\node (v5)  at (0,6)   [var] {};
\node (v6)  at (0,5)   [var] {};
\node (v7)  at (0,4)   [var] {};
\node (v8)  at (0,3)   [var] {};
\node (v9)  at (0,2)   [var] {};
\node (v10) at (0,1)   [var] {};

\node (p1) at (1.5,9)   [par] {};
\node (p2) at (1.5,7)   [par] {};
\node (p3) at (1.5,5.5) [par] {};
\node (p4) at (1.5,3.5) [par] {};
\node (p5) at (1.5,1.5) [par] {};

\node at   (2.,9)   [txt] {$j_1$};
\node at   (2.,7)   [txt] {$j_2$};
\node at   (2.,5.5) [txt] {$j_3$};
\node at   (2.,3.5) [txt] {$j_4$};
\node at   (2.,1.5) [txt] {$j_5$};

\node at   (3,8)  [txt6] {$\}$};
\node at   (3,4.5)  [txt6] {$\}$};

\node (m) at   (4.2,6.7) [txt4] {$\rightarrow$};

\draw[-] (v1) -- (p1);
\draw[-] (v2) -- (p1);
\draw[-] (v3) -- (p1);

\draw[-] (v2) -- (p2);
\draw[-] (v3) -- (p2);
\draw[-] (v4) -- (p2);

\draw[-] (v4) -- (p3);
\draw[-] (v5) -- (p3);
\draw[-] (v6) -- (p3);

\draw[-] (v6) -- (p4);
\draw[-] (v7) -- (p4);
\draw[-] (v8) -- (p4);

\draw[-] (v6)  -- (p5);
\draw[-] (v9)  -- (p5);
\draw[-] (v10) -- (p5);


\node (vp1)  at (5.5,10)  [var] {};
\node (vp2)  at (5.5,9)   [var] {};
\node (vp3)  at (5.5,8)   [var] {};
\node (vp4)  at (5.5,7)   [var] {};
\node (vp5)  at (5.5,6)   [var] {};
\node (vp6)  at (5.5,5)   [var] {};
\node (vp7)  at (5.5,4)   [var] {};
\node (vp8)  at (5.5,3)   [var] {};
\node (vp9)  at (5.5,2)   [var] {};
\node (vp10) at (5.5,1)   [var] {};

\node (pp1) at (7,8)   [par] {};
\node (pp2) at (7,4.5) [par] {};
\node (pp3) at (7,1.5) [par] {};

\node at   (7.5,8)   [txt] {$j'_1$};
\node at   (7.5,4.5) [txt] {$j'_2$};
\node at   (7.5,1.5) [txt] {$j'_3$};

\draw[-] (vp1) -- (pp1);
\draw[-] (vp2) -- (pp1);
\draw[-] (vp3) -- (pp1);
\draw[-] (vp4) -- (pp1);

\draw[-] (vp4) -- (pp2);
\draw[-] (vp5) -- (pp2);
\draw[-] (vp6) -- (pp2);
\draw[-] (vp7) -- (pp2);
\draw[-] (vp8) -- (pp2);

\draw[-] (vp6)  -- (pp3);
\draw[-] (vp9)  -- (pp3);
\draw[-] (vp10) -- (pp3);

\end{tikzpicture}
\end{center}
\caption{Starting from the graph $\cG$ on the left part of the figure that represents the code $\cC$, we merge constraint nodes, thus creating a new representation of $\cC$ with an underlying graph $\cG'$ shown on the right.
\label{fig:merge_example}}
\end{figure}

Denote by $\cJ'$ the new set of constraint nodes corresponding to $\cG'$. Now, we can apply the LP relaxation to the new representation and calculate $\hbc'$ defined by,
$$
\hbc' = \argmin_{\bc \in \cQ'(\cC)} \P(\bc)
$$
where
\begin{multline*}
\cQ'(\cC) = \left\{ \bc \: : \: \exists \bomega = \left\{ w'_{j,\bg}
\right\}_{j\in\cJ',\: \bg\in\cC_j'} \mbox{ s.t.
\eqref{eq:positive_w}, \eqref{eq:w_sum_1}  } \right.
\\ \left. \mbox{ and \eqref{eq:ci} (with $\cJ',\cC_j'$ instead of $\cJ,\cC_j$) hold} \right\}
\end{multline*}

\begin{proposition}
\label{prop:Q2_subset_Q1}
$\cQ'(\cC) \subseteq \cQ(\cC)$.
\end{proposition}

\beginproof
Suppose that $\bc \in \cQ'(\cC)$. We need to show that $\bc \in
\cQ(\cC)$. First suppose that the only difference between $\cG$ and
$\cG'$ is that the constraint nodes $j_1, j_2, \ldots, j_r \in \cJ$
in $\cG$ are merged to form the constraint node $j' \in \cJ'$ in
$\cG'$. All the other constraint nodes in $\cG$ and $\cG'$ are
identical (i.e. they are associated with the same constituent codes
and are connected to the same variable nodes). Now, there exists
$\left\{ w'_{j,\bg} \right\}_{j\in\cJ',\: \bg\in\cC_j'}$ such
that~\eqref{eq:positive_w}, \eqref{eq:w_sum_1} and \eqref{eq:ci}
(with $\cJ',\cC_j'$ instead of $\cJ,\cC_j$) hold. For $j = j_1, j_2,
\ldots, j_r$, define $w_{j,\bg}$ by,
$$
w_{j,\bg} \defined \sum_{\bg'\in \cC_j',\: \bg'_{\cN_j}=\bg}
w'_{j',\bg'}
$$
where $\bg'_{\cN_j}$ is the sub-codeword $\bg'$ restricted to
indices $i\in\cN_j$. For all other values of $j\in\cJ$ we define
$w_{j,\bg} = w'_{j'(j),\bg}$ where $j'(j)\in\cJ'$ is the constraint
node in $\cG'$ that corresponds to the constraint node $j\in\cJ$ in
$\cG$.

Note that $\bg'\in \cC_j'$ only if $\bg'_{\cN_j} \in \cC_j$ (since
any codeword associated with the merged node $j'$ is a codeword of
$\cC_j$ for all $j=j_1,\ldots,j_r$). It follows that for all
$j=j_1,\ldots,j_r$,
$$
w_{j,\bg} \ge 0
$$
$$
\sum_{\bg \in \cC_j} w_{j,\bg} = \sum_{\bg' \in \cC_j'} w'_{j',\bg'}
= 1
$$
$$
\sum_{\bg\in \cC_j \: , \: g_i=1} w_{j,\bg} = \sum_{\bg'\in \cC_j'
\: , \: g'_i=1} w'_{j',\bg'} = c_i \qquad \forall i \in \cN_j
$$
Hence we conclude that $\bc \in \cQ(\cC)$. We have therefore shown
that in this case $\cQ'(\cC) \subseteq \cQ(\cC)$.

Now suppose that $\cG'$ is formed by a general merge up of
constraint nodes in $\cG$. In this case the process of constructing
$\cG'$ from $\cG$ can be separated to a series of merge up steps of
constraint nodes, such that at each step we merge a single set of
parity check nodes from the graph constructed so far in order to
form a new graph. That is, starting from $\cG=\tilde{\cG}_1$, we
first construct $\tilde{\cG}_2$ with the polytope $\tilde{\cQ}_2$.
Then we construct $\tilde{\cG}_3$ with the polytope $\tilde{\cQ}_3$
and so forth, until we obtain $\tilde{\cG}_s=\cG'$ with the polytope
$\tilde{\cQ}_s=\cQ'(\cC)$. For $i=1,2,\ldots,s-1$, the graph
$\tilde{\cG}_{i+1}$ is obtained from $\tilde{\cG}_{i}$ by merging a
single set of constraint nodes into one constraint node and leaving
all the other constraint nodes intact (e.g., in the example shown in
Figure~\ref{fig:merge_example}, $\tilde{\cG}_1=\cG$ is the graph
shown on the left. We first merge $j_1$ and $j_2$ into $j'_1$, thus
forming $\tilde{\cG}_2$. Then we merge $j_3$ and $j_4$ into $j'_2$
thus forming $\tilde{\cG}_3=\cG'$, which is the graph shown on the
right). By what we have shown so far we know that,
$$
\cQ'(\cC) = \tilde{\cQ}_s \subseteq \tilde{\cQ}_{s-1} \subseteq \ldots \subseteq \tilde{\cQ}_1 = \cQ(\cC)
$$
\finproof

Note that when we merge all the constraint nodes into one constraint
node, the resulting polytope, $\cQ^{\rm ML}(\cC)$, is the convex
combination of all possible codewords. That is, the set of vertices
of this polytope is the set of codewords of $\cC$, and the
minimization of $\P(\bc)$ over $\bc \in \cQ^{\rm ML}(\cC)$ yields
the ML decoded codeword. Thus we have $\cQ^{\rm ML}(\cC) \subseteq
\cQ'(\cC) \subseteq \cQ(\cC)$. Note also that for any codeword $\bc$
we have $\bc \in \cQ(\cC)$, $\bc \in \cQ'(\cC)$ and $\bc \in
\cQ^{\rm ML}(\cC)$. That is, $\cQ'(\cC)$ is a finer relaxation of
$\cQ^{\rm ML}(\cC)$ than $\cQ(\cC)$.

Let the error probability $P_e(\bc)$ of the decoder that uses the
initial graph $\cG$, given that the transmitted codeword is $\bc$,
be defined by,
$$
P_e(\bc) = \Pr \left( \exists \btc\in\cQ(\cC) \: : \: \btc \ne \bc, \: \P(\btc) \le \P(\bc) \right)
$$
Similarly, the error probability $P_e(\bc)$ of the decoder that uses the graph $\cG'$, given that the transmitted codeword is $\bc$, is defined by,
$$
P_e'(\bc) = \Pr \left( \exists \btc\in\cQ'(\cC) \: : \: \btc \ne
\bc, \: \P(\btc) \le \P(\bc) \right)
$$

By proposition~\ref{prop:Q2_subset_Q1} we immediately obtain the following,

\begin{proposition}
\label{prop:2_improves_1}
$P_e'(\bc) \le P_e(\bc)$.
\end{proposition}
This proposition holds trivially due to the fact that the second LP
decoder uses a tighter relaxation of $\cQ^{\rm ML}(\cC)$.

\begin{proposition}
\label{prop:2_equals_1} Suppose that in the merging process of
constraint nodes we only merge constraint nodes $j_1,\ldots,j_r$
such that the subgraph of $\cG$ induced by $j_1,\ldots,j_r$ and
their direct neighbors is cycle free. Then $\cQ(\cC) = \cQ'(\cC)$
and $P_e(\bc) = P_e'(\bc)$.
\end{proposition}

\beginproof
We already know from Proposition~\ref{prop:Q2_subset_Q1} that
$\cQ'(\cC) \subseteq \cQ(\cC)$. Thus we only need to show that
$\cQ(\cC) \subseteq \cQ'(\cC)$. We first prove this assertion for
the case where $\cG'$ is constructed from $\cG$ by merging only a
single pair of constraint nodes, $j_1, j_2 \in \cJ$, of $\cG$. In
this case, by the assumption of the proposition, there are two
possibilities.
\begin{enumerate}
\item
There is no overlap between the variable nodes that are connected to $j_1$ and the variable nodes that are connected to $j_2$. Denote by $\bc_1$ and $\bc_2$ the sub-codewords associated with the variable nodes connected to $j_1$ and $j_2$ respectively. In this case there is no overlap between $\bc_1$ and $\bc_2$. This case is illustrated in Figure~\ref{fig:merge_zero_overlap}.
\item
Exactly one variable node is connected to both $j_1$ and $j_2$, that is there is an overlap of one bit, $c_{1,2}$, between $\bc_1$ and $\bc_2$. This case is illustrated in Figure~\ref{fig:merge_one_overlap}.
\end{enumerate}

\begin{figure}[hbtp]
\begin{center}
\begin{tikzpicture}
[var/.style={circle,draw=black!50,fill=red!40},
par/.style={rectangle,draw=black!50,fill=blue!40},
txt/.style={},
txt2/.style={scale=2},
txt3/.style={scale=3},
txt4/.style={scale=4},
txt5/.style={scale=5},
txt6/.style={scale=6},
]

\node at   (0,7) [txt2] {$\vdots$};
\node (v1) at (0,6) [var] {};
\node (v2) at (0,5) [var] {};
\node (v3) at (0,4) [var] {};
\node (v4) at (0,3) [var] {};
\node (v5) at (0,2) [var] {};
\node (v6) at (0,1) [var] {};
\node at   (0,0.5) [txt2] {$\vdots$};

\node at   (1.5,6) [txt2] {$\vdots$};
\node (p1) at (1.5,5) [par] {};
\node (p2) at (1.5,2) [par] {};
\node at   (1.5,1.5) [txt2] {$\vdots$};

\node at   (2.,5) [txt] {$j_1$};
\node at   (2.,2) [txt] {$j_2$};

\node (m) at   (4,3.5) [txt] {merge and create $j'$};
\draw[<-]  (2.3,4.8) -- (m);
\draw[<-]  (2.3,2.2) -- (m);

\draw[-] (v6) -- (p2);
\draw[-] (v5) -- (p2);
\draw[-] (v4) -- (p2);

\draw[-] (v3) -- (p1);
\draw[-] (v2) -- (p1);
\draw[-] (v1) -- (p1);

\node at   (-1.5,5) [txt] {$\bc_1$};
\node at   (-.8,5)  [txt6] {$\{$};

\node at   (-1.5,2) [txt] {$\bc_2$};
\node at   (-.8,2)  [txt6] {$\{$};

\end{tikzpicture}
\end{center}
\caption{Merging two parity nodes, $j_1$ and $j_2$, in the Tanner graph of the code $\cC$. In this case there is no overlap between $\bc_1$ and $\bc_2$.
\label{fig:merge_zero_overlap}}
\end{figure}

\begin{figure}[hbtp]
\begin{center}
\begin{tikzpicture}
[var/.style={circle,draw=black!50,fill=red!40},
par/.style={rectangle,draw=black!50,fill=blue!40},
txt/.style={},
txt2/.style={scale=2},
txt3/.style={scale=3},
txt4/.style={scale=4},
txt5/.style={scale=5},
txt6/.style={scale=6},
]

\node at   (0,6.) [txt2] {$\vdots$};
\node (v1) at (0,5) [var] {};
\node (v2) at (0,4) [var] {};
\node (v3) at (0,3) [var] {};
\node (v4) at (0,2) [var] {};
\node (v5) at (0,1) [var] {};
\node at   (0,0.5) [txt2] {$\vdots$};

\node at   (1.5,5.) [txt2] {$\vdots$};
\node (p1) at (1.5,4) [par] {};
\node (p2) at (1.5,2) [par] {};
\node at   (1.5,1.5) [txt2] {$\vdots$};

\node at   (2.,4) [txt] {$j_1$};
\node at   (2.,2) [txt] {$j_2$};

\node (m) at   (4,3) [txt] {merge and create $j'$};
\draw[<-]  (2.3,3.8) -- (m);
\draw[<-]  (2.3,2.2) -- (m);

\draw[-] (v1) -- (p1);
\draw[-] (v2) -- (p1);
\draw[-] (v3) -- (p1);

\draw[-] (v3) -- (p2);
\draw[-] (v4) -- (p2);
\draw[-] (v5) -- (p2);

\node at   (-1.5,4) [txt] {$\bc_1$};
\node at   (-.8,4)  [txt6] {$\{$};

\node at   (-1.5,2) [txt] {$\bc_2$};
\node at   (-.8,2)  [txt6] {$\{$};

\node at   (0.5,2.95) [txt] {$c_{1,2}$};

\end{tikzpicture}
\end{center}
\caption{Merging two parity nodes, $j_1$ and $j_2$, in the Tanner graph of the code $\cC$. In this case there is an overlap of one bit, $c_{1,2}$, between $\bc_1$ and $\bc_2$.
\label{fig:merge_one_overlap}}
\end{figure}
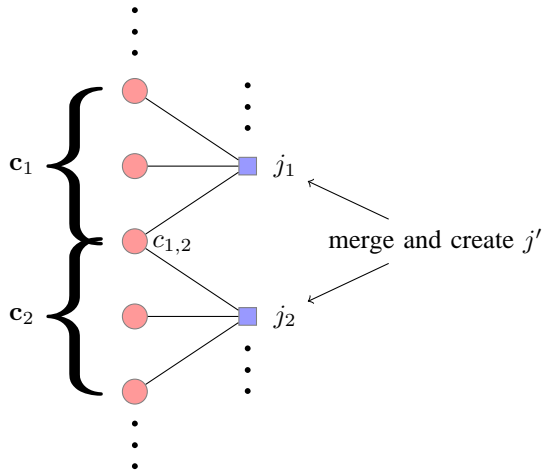

Now consider the first case (Figure~\ref{fig:merge_zero_overlap}). Suppose that $\bc \in \cQ(\cC)$. Then for $j=j_1, j_2$ there exists a set of weights $\{ w_{j,\bg} \ge 0 \}$ for $\bg\in \cC_j$ such that
$$
\sum_{\bg \in \cC_j} w_{j,\bg} = 1
$$
and
$$
c_i = \sum_{\bg\in \cC_j \: , \: g_i=1} w_{j,\bg} \qquad \forall \: i \in \cN_j
$$
Define the set $\{ w_{j',\bg} \}$ for $\bg=(\bg_1,\bg_2)$, where $\bg_1$ and $\bg_2$ correspond to the sub-codewords $\bc_1$ and $\bc_2$ respectively in Figure~\ref{fig:merge_zero_overlap}, by
$$
w_{j',(\bg_1,\bg_2)} \defined w_{j_1,\bg_1} \cdot w_{j_2,\bg_2}
$$
It follows from this definition that $w_{j',(\bg_1,\bg_2)} \ge 0$, and
$$
\sum_{\bg \in \cC_j'} w_{j',\bg} = \sum_{\bg_1 \in \cC_{j_1}}
\sum_{\bg_2 \in \cC_{j_2}} w_{j_1,\bg_1} w_{j_2,\bg_2} = 1
$$
Furthermore, for $i\in\cN_{j_1}$
\begin{multline*}
\sum_{\substack{\bg\in \cC_j' \\ g_i=1}} w_{j',\bg} =
\sum_{\substack{(\bg_1,\bg_2) \in \cC_j' \\ g_{1,i}=1}}
w_{j',(\bg_1,\bg_2)} \\ = \sum_{\substack{\bg_1 \in \cC_{j_1} \\
g_{1,i}=1}} w_{j_1,\bg_1} \sum_{\bg_2 \in \cC_{j_2}} w_{j_2,\bg_2} =
c_i
\end{multline*}
where $g_{1,i}$ ($g_{2,i}$, respectively) is the $i$'th component of $g_1$ ($g_2$). Similarly, for $i\in\cN_{j_2}$
$$
\sum_{\bg\in \cC_j' \: , \: g_i=1} w_{j',\bg} = c_i
$$
This shows that in the case considered $\bc \in \cQ'(\cC)$.

Next consider the second case (Figure~\ref{fig:merge_one_overlap}). Suppose that $\bc \in \cQ(\cC)$. Then for $j=j_1, j_2$ there exists a set of weights $\{ w_{j,\bg} \ge 0 \}$ for $\bg\in \cC_j$ such that
\bre
\sum_{\bg \in \cC_j} w_{j,\bg} = 1
\label{eq:w_sum_1_oo}
\ere
and
\bre
c_i = \sum_{\bg\in \cC_j \: , \: g_i=1} w_{j,\bg} \qquad \forall \: i \in \cN_j
\label{eq:ci_oo}
\ere
Let $\hbc_1$ denote the sub-codeword $\bc_1$ excluding the common bit $c_{1,2}$. Similarly $\hbc_2$ is the sub-codeword $\bc_2$ excluding $c_{1,2}$.
Define the set $\{ w_{j',\bg} \}$ for $\bg=(\hbg_1,g_{1,2},\hbg_2)$, where $\hbg_1$, $g_{1,2}$ and $\hbg_2$ correspond to $\hbc_1$, $c_{1,2}$ and $\hbc_2$ respectively, by
\begin{multline*}
w_{j',\bg} = w_{j',(\hbg_1,g_{1,2},\hbg_2)} \defined \\ \left\{
  \begin{array}{ll}
    w_{j_1,(\hbg_1,1)} \cdot w_{j_2,(1,\hbg_2)} / c_{1,2} & \hbox{if $g_{1,2}=1$} \\
    w_{j_1,(\hbg_1,0)} \cdot w_{j_2,(0,\hbg_2)} / (1-c_{1,2}) & \hbox{if $g_{1,2}=0$}
  \end{array}
\right.
\end{multline*}
It follows from this definition that
$w_{j',\bg} \ge 0$, and
\begin{align*}
&\sum_{\bg \in \cC_j'} w_{j',\bg} \\ &=\frac{1}{1-c_{1,2}}
\sum_{\hbg_1 \: : \: (\hbg_1,0) \in \cC_{j_1}} \sum_{\hbg_2 \: : \:
(0,\hbg_2) \in \cC_{j_2}} w_{j_1,(\hbg_1,0)} w_{j_2,(0,\hbg_2)}
\\
&+ \frac{1}{c_{1,2}} \sum_{\hbg_1 \: : \: (\hbg_1,1) \in \cC_{j_1}}
\sum_{\hbg_2 \: : \: (1,\hbg_2) \in \cC_{j_2}} w_{j_1,(\hbg_1,1)}
w_{j_2,(1,\hbg_2)} \\ &= 1
\end{align*}
where the last equality follows from the following equalities, which
are implied by~\eqref{eq:w_sum_1_oo} and~\eqref{eq:ci_oo},
\begin{align}
c_{1,2}
&=
\sum_{\hbg_1 \: : \: (\hbg_1,1) \in \cC_{j_1}} w_{j_1,(\hbg_1,1)}
=
\sum_{\hbg_2 \: : \: (1,\hbg_2) \in \cC_{j_2}} w_{j_2,(1,\hbg_2)}
\label{eq:c12} \\
1-c_{1,2}
&=
\sum_{\hbg_1 \: : \: (\hbg_1,0) \in \cC_{j_1}} w_{j_1,(\hbg_1,0)}
=
\sum_{\hbg_2 \: : \: (0,\hbg_2) \in \cC_{j_2}} w_{j_2,(0,\hbg_2)}
\label{eq:1mc12}
\end{align}

It remains to show that \bre \label{eq:eqtoci} \sum_{\bg\in \cC_j'
\: , \: g_i=1} w_{j',\bg} = c_i \ere This is shown by considering
the following three possibilities for $i$: $i\in\cN_{j_1} \cap
\cN_{j_2}$, $i \in \cN_{j_1} \setminus \cN_{j_2}$ and $i \in
\cN_{j_2} \setminus \cN_{j_1}$. In the first case, $i\in\cN_{j_1}
\cap \cN_{j_2}$, we have $c_i = c_{1,2}$. In this case
\begin{align*}
\sum_{\substack{\bg\in \cC_j' \\ g_i=1}} w_{j',\bg} &=
\sum_{(\hbg_1,1,\hbg_2) \in \cC_j'} \frac{1}{c_{1,2}}
w_{j_1,(\hbg_1,1)} \cdot w_{j_2,(1,\hbg_2)}
\\ &=
\frac{1}{c_{1,2}} \sum_{\hbg_1 \: : \: (\hbg_1,1) \in \cC_{j_1}} w_{j_1,(\hbg_1,1)} \sum_{\hbg_2 \: : \:  (1,\hbg_2) \in \cC_{j_2}} w_{j_2,(1,\hbg_2)}
\\ &=
c_{1,2}
\end{align*}
where the last equality follows from~\eqref{eq:c12}.

In the second case, $i \in \cN_{j_1} \setminus \cN_{j_2}$, we have
\begin{align*}
\sum_{\substack{\bg\in \cC_j' \\ g_i=1}} w_{j',\bg} &=
\sum_{\substack{(\hbg_1,1,\hbg_2) \in \cC_j' \\ g_{1,i}=1}}
\frac{w_{j_1,(\hbg_1,1)} \cdot w_{j_2,(1,\hbg_2)}}{c_{1,2}}
\\ &+
\sum_{\substack{(\hbg_1,0,\hbg_2) \in \cC_j' \\ g_{1,i}=1}}
\frac{w_{j_1,(\hbg_1,0)} \cdot w_{j_2,(0,\hbg_2)}}{1-c_{1,2}}
\\ &=
\sum_{\substack{\hbg_1 \: : \: (\hbg_1,1) \in \cC_{j_1}
\\ g_{1,i}=1}} \frac{w_{j_1,(\hbg_1,1)}}{c_{1,2}}
\sum_{\hbg_2 \: : \: (1,\hbg_2) \in \cC_{j_2}} w_{j_2,(1,\hbg_2)}
\\ &+
\sum_{\substack{\hbg_1 \: : \: (\hbg_1,0) \in \cC_{j_1}
\\ g_{1,i}=1}} \frac{w_{j_1,(\hbg_1,0)}}{1-c_{1,2}}
\sum_{\hbg_2 \: : \: (0,\hbg_2) \in \cC_{j_2}} w_{j_2,(0,\hbg_2)}
\\ &=
\sum_{\substack{\hbg_1 \: : \: (\hbg_1,1) \in \cC_{j_1} \\
g_{1,i}=1}} w_{j_1,(\hbg_1,1)} + \sum_{\substack{\hbg_1 \: : \:
(\hbg_1,0) \in \cC_{j_1} \\ g_{1,i}=1}} w_{j_1,(\hbg_1,0)}
\\ &= c_i
\end{align*}
where the third equality is due to~\eqref{eq:c12} and~\eqref{eq:1mc12}, and the fourth equality is due to~\eqref{eq:ci_oo} (with $j=j_1$).
The proof that~\eqref{eq:eqtoci} holds in the third case, $i \in \cN_{j_2} \setminus \cN_{j_1}$, is identical. Thus in all three cases considered~\eqref{eq:eqtoci} holds, and we conclude that $\bc \in \cQ'(\cC)$ as claimed.

We thus conclude that when we merge a single pair of constraint
nodes, $\cQ(\cC) = \cQ'(\cC)$. We proceed to prove that this holds
for a general merge up of a set of constraint nodes
$V=\{j_1,\ldots,j_r\}$, for which the subgraph of $\cG$ induced by
the nodes in $V$ and their direct neighbors is cycle free, into one
node $j'$. We denote the new graph by $\cG'$. We start from the
original graph $\cG$ with the corresponding polytope
$\cQ(\cC)=\cQ^{(1)}(\cC)$, and select a constraint node
$\tilde{j}\in V$ arbitrarily. Now we select a node to merge with
node $\tilde{j}$. If there exists some node $\tilde{j}'\in V$ such
that $\cN_{\tilde{j}} \cap \cN_{\tilde{j}'} \neq \emptyset$ then
$\tilde{j}'$ is selected to be merged with $\tilde{j}$. Otherwise,
$\tilde{j}'\in V$ is picked arbitrarily. Now $\tilde{j}$ and
$\tilde{j}'$ are merged into a new constraint node
$\tilde{j}^{(1)}$. After performing the merging operation on these
two nodes, a new graph $\cG^{(2)}$ is formed with the corresponding
polytope, $\cQ^{(2)}(\cC)$. Now, in $\cG$, the two constraint nodes
selected for merging either have no common variable node or are
connected to a single common variable node. By what we have already
shown, we know that $\cQ^{(1)}(\cC) = \cQ^{(2)}(\cC)$. Now we repeat
the process. If there is a node $\tilde{j}''$ within the remaining
nodes in $V$ such that $\cN_{\tilde{j}^{(1)}} \cap \cN_{\tilde{j}''}
\neq \emptyset$ then $\tilde{j}''$ is selected to be merged with
$\tilde{j}^{(1)}$, otherwise $\tilde{j}''$ is selected arbitrarily.
Once $\tilde{j}''$ is chosen, it is merged with $\tilde{j}^{(1)}$
into a new constraint node $\tilde{j}^{(2)}$, forming a new graph
$\cG^{(3)}$ with its corresponding polytope, $\cQ^{(3)}(\cC)$. By
this selection process, in $\cG^{(2)}$, the two constraint nodes
selected for merging again have either no common variable node or
are connected to a single common variable node, and thus
$\cQ^{(2)}(\cC) = \cQ^{(3)}(\cC)$. This merge up process continues
until we have created the graph $\cG' = \cG^{(r)}$ with polytope
$\cQ'(\cC) = \cQ^{(r)}(\cC)$. Thus we have shown that,
$$
\cQ{(\cC)} = \cQ^{(1)}(\cC) = \cQ^{(2)}(\cC) = \ldots =
\cQ^{(r)}(\cC) = \cQ'(\cC)
$$

The proof of the general case, where we merge several
non-overlapping groups of variable nodes in $\cG$ in order to form
$\cG'$ (e.g., in the example shown in
Figure~\ref{fig:merge_example}, where we merge $j_1$ and $j_2$ into
$j'_1$, and $j_3$ and $j_4$ into $j'_2$), now follows immediately.
Thus we have shown that $\cQ(\cC) = \cQ'(\cC)$. This implies that
$P_e(\bc) = P_e'(\bc)$. \finproof

Proposition~\ref{prop:2_equals_1} implies the following. Consider
some graphical representation $\cG$ of a code $\cC$ and suppose that
we wish to obtain a new graphical representation $\cG'$ by merging
constraint nodes such that when using LP decoding, the error
probability of $\cG'$ is smaller than the error probability of
$\cG$. Then, when considering candidate constraint nodes for merge
up, it is sufficient to consider only constraint nodes
$j_1,\ldots,j_r$ such that the subgraph of $\cG$ induced by
$j_1,\ldots,j_r$ and their direct neighbors contains a cycle.

In the sequel, we will use a variant of
Proposition~\ref{prop:2_equals_1}, where upon merging nodes
$j_1,j_2,\dots,j_r$ into $j'$, we keep the original nodes
$j_1,j_2,\dots,j_r$ as well as the new node $j'$ (the reason for
this is that we want to enable some of the nodes $j_1,j_2,\dots,j_r$
to be merged with other nodes). Due to
Propositions~\ref{prop:Q2_subset_Q1} and~\ref{prop:2_equals_1},
keeping $j_1,j_2,\dots,j_r$ has no effect on the resulting
$\cQ'(\cC)$ and $P_e'(\bc)$.

\section{Bound on the Minimum and Fractional Distance of Specific Codes} \label{sec:min distance bound}
It was already shown in \cite{lpdecode} that LP decoding can be used
to obtain the fractional distance, which is also a lower bound on
the minimum distance of specific linear codes, in polynomial time
complexity. In this section, we show how Algorithm~\ref{alg} can be
used to obtain a lower bound on the minimum distance of a given
code, and that this bound is also an upper bound on the fractional
distance. Assuming the degree of the check nodes is bounded by a
constant independent of $N$, this procedure requires execution of
Algorithm~\ref{alg} a number of times proportional to $N$, and thus
the total complexity is $O(N^2)$.

Let $r \in \cJ$ be a check node and let $\bc_r$ be some nonzero
local codeword of $r$. The idea is to search for a minimum-weight
vector subject to the constraint that $\bc_r$ is the local codeword
on node $r$ (later we will show how such vectors allow to obtain a
lower bound on the minimum distance). This conforms to the setting
of Problem-P, if we set $\gamma_i=1$ for all $i$. Denote $\cI_r
\triangleq \cI \backslash \cN_r$ and $\cJ_r \triangleq \cJ
\backslash \{r\}$. We would like to obtain the solution by running
Algorithm~\ref{alg} on a modified version of the graph in which we
remove the node $r$, all its neighbors $\cN_r$ and the edges
connected to these neighbors. However, in order to maintain a
correct representation of the code, we must account for the specific
nonzero local codeword $\bc_r$. To do this, we look at all
constraint nodes $j \in \cJ_r$ where $\cN_j \cap \cN_r \neq
\emptyset$. Let $j \in \cJ_r$ be such a constraint node described by
a matrix $H_j$ with columns $\{\bh_i\}_{i \in \cN_j}$; this
situation is exemplified in Figure~\ref{fig: coset codes a}. Since
we will be forcing the variables in $\cN_j \bigcap \cN_r$ to a value
depending on $\bc_r$, we have that the local codeword vector $\bc_j
\in \cC_j$ on node $j$ satisfies \bre \label{eq: h_r,S0} \sum_{i \in
\cN_j \setminus \cN_r} c_i \bh_i= \sum_{i \in \cN_j \bigcap \cN_r}
c_i \bh_i \triangleq \tilde{\bh}_j^{r,\bc_r} \ere We now remove
parity check $r$ and all its neighbors from the graph. Consequently,
in the remaining graph we observe that check node $j$ describes a
coset code where a standard constraint of the form $H_j\bc_j=0$ is
replaced by the constraint $H_j^r\bc_j^r=\tilde{\bh}_j^{r,\bc_r}$,
where $H^r_j=\{\bh_i\}_{i \in \cN_j \backslash \cN_r}$ and
$\bc^r_j=\{c_i\}_{i \in \cN_j \backslash \cN_r}$. This is
exemplified in Figure~\ref{fig: coset codes b}. Denote the set of
local codewords of this coset code by $\cC_j^{r,\bc_r}$. Also denote
$\bc^{(r)}=\{c_i\}_{i\in \cI_r}$,
$\bomega^{(r)}=\{\omega_{j,\bg}\}_{j\in \cJ_r, \bg \in
\cC_j^{r,\bc_r}}$. Problem-$\text{P}^{r,\bc_r}$ is defined as
follows. \bre \min_{\bc^{(r)},\bomega^{(r)}} \sum_{i\in\cI_r} c_i
\label{eq:ml_lp} \ere subject to \bre w_{j,\bg} \ge 0 \qquad \forall
j \in \cJ_r \: , \: \bg \in \cC_j^{r,\bc_r}
\label{eq:positive_w_reduced} \ere \bre \sum_{\bg \in
\cC_j^{r,\bc_r}} w_{j,\bg} = 1 \qquad \forall j \in \cJ_r
\label{eq:w_sum_1_reduced} \ere \bre c_i = \sum_{\bg\in
\cC_j^{r,\bc_r} \: , \: g_i=1} w_{j,\bg} \qquad \forall j\in\cJ_r \:
, \: i \in \cN_j \setminus \cN_r \: . \label{eq: problem p_rs0,4}
\ere

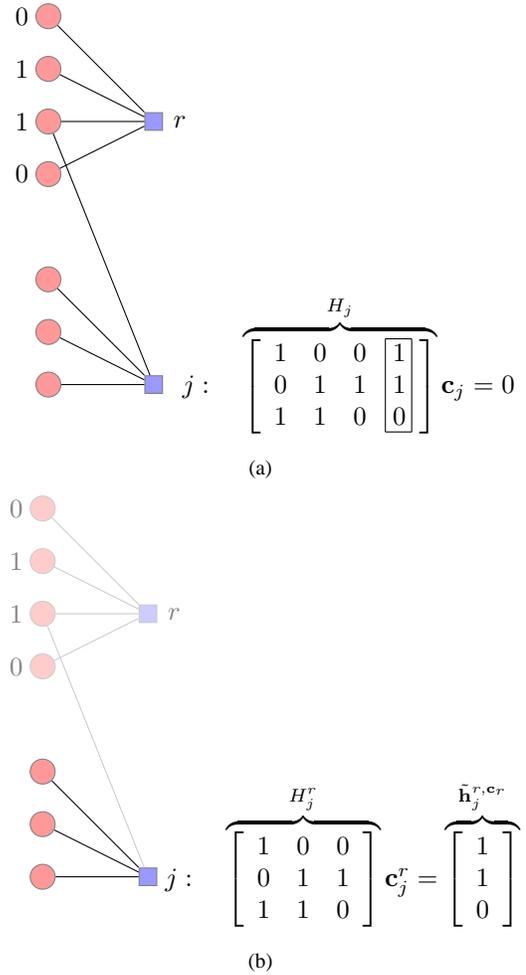
\begin{figure}[hbtp] \label{fig: coset codes}
\begin{center}
\subfigure[]{ \label{fig: coset codes a}
\begin{tikzpicture}
[var/.style={circle,draw=black!50,fill=red!40},
varg/.style={circle,draw=black!10,fill=red!10},
parg/.style={circle,draw=black!10,fill=red!10},
par/.style={rectangle,draw=black!50,fill=blue!40}, txt/.style={},
txt2/.style={scale=2}, txt3/.style={scale=3}, txt4/.style={scale=4},
txt5/.style={scale=5}, txt6/.style={scale=6}, scale=0.7]

\node (vp1)  at (5,10)  [var] {}; \node (vp2)  at (5,9)   [var] {};
\node (vp3)  at (5,8)   [var] {}; \node (vp4)  at (5,7)   [var] {};
\node (vp5)  at (5,5)   [var] {}; \node (vp6)  at (5,4)   [var] {};
\node (vp7)  at (5,3)   [var] {}; \node (pp1) at (7,8)   [par] {};
\node (pp2) at (7,3) [par] {};

\node at (4.5,10)   [txt] {$0$}; \node at (4.5,9)   [txt] {$1$};
\node at (4.5,8)   [txt] {$1$}; \node at (4.5,7)   [txt] {$0$};
\node at (7.5,8)   [txt] {$r$}; \node at (10.7,3.37) [txt] {$j:
\quad
\overbrace{\left[\begin{array}{l l l l} 1 & 0 & 0 & 1 \\ 0 & 1 & 1 & 1 \\
1 & 1 & 0 & 0\end{array}\right]}^{H_j}\bc_j=0$};

\draw[-] (11.4,3.87) -- (11.9,3.87); \draw[-] (11.4,2.1) --
(11.9,2.1); \draw[-] (11.4,3.87) -- (11.4,2.1); \draw[-] (11.9,3.87)
-- (11.9,2.1); \draw[-] (vp1) -- (pp1); \draw[-] (vp2) -- (pp1);
\draw[-] (vp3) -- (pp1); \draw[-] (vp4) -- (pp1);

\draw[-] (vp3) -- (pp2); \draw[-] (vp5) -- (pp2); \draw[-] (vp6) --
(pp2); \draw[-] (vp7) -- (pp2);
\end{tikzpicture} }
\subfigure[]{ \label{fig: coset codes b}
\begin{tikzpicture}
[var/.style={circle,draw=black!50,fill=red!40},
varg/.style={circle,draw=black!20,fill=red!20},
parg/.style={rectangle,draw=black!20,fill=blue!20},
weakline/.style={draw=black!20},
par/.style={rectangle,draw=black!50,fill=blue!40}, txt/.style={},
txtweak/.style={color=gray},  txt2/.style={scale=2},
txt3/.style={scale=3}, txt4/.style={scale=4}, txt5/.style={scale=5},
txt6/.style={scale=6}, scale=0.7]

\node (vp1)  at (5,10)  [varg] {}; \node (vp2)  at (5,9)   [varg]
{}; \node (vp3)  at (5,8)   [varg] {}; \node (vp4)  at (5,7)
[varg] {}; \node (vp5)  at (5,5)   [var] {}; \node (vp6)  at (5,4)
[var] {}; \node (vp7)  at (5,3)   [var] {}; \node (pp1) at (7,8)
[parg] {}; \node (pp2) at (7,3) [par] {};

\node at (4.5,10)   [txtweak] {$0$}; \node at (4.5,9)   [txtweak]
{$1$}; \node at (4.5,8)   [txtweak] {$1$}; \node at (4.5,7)
[txtweak] {$0$}; \node at (7.5,8)   [txtweak] {$r$}; \node at
(10.7,3.43) [txt] {$j: \quad \overbrace{\left[\begin{array}{l l l} 1
& 0 & 0 \\ 0 & 1 & 1
\\ 1 & 1 & 0
\end{array}\right]}^{H^r_j}\bc^r_j=\overbrace{\left[\begin{array}{l} 1
\\ 1 \\ 0 \end{array} \right]}^{\tilde{\bh}_j^{r,\bc_r}}$};

\draw[weakline] (vp1) -- (pp1); \draw[weakline] (vp2) -- (pp1);
\draw[weakline] (vp3) -- (pp1); \draw[weakline] (vp4) -- (pp1);

\draw[weakline] (vp3) -- (pp2); \draw[-] (vp5) -- (pp2); \draw[-]
(vp6) -- (pp2); \draw[-] (vp7) -- (pp2); 

\end{tikzpicture} }
\end{center}
\caption{Coset codes formed when a codeword $\bc_r=(0,1,1,0)$ is
forced on constraint node $r$ which is removed. (a) The rightmost
column of $H_j$ corresponds to a removed variable node forced to a
'$1$' value. (b) The new code $\cC_j^{r,\bc_r}$ is formed as $r$ and
$\cN_r$ are removed from the graph.}
\end{figure}


Thus, by setting $\gamma_i=1$ $\forall i \in \cI$,
Problem-$\text{P}^{r,\bc_r}$ is easily seen to be an instance of
Problem-P, and thus its solution can be approximated arbitrarily
closely by Algorithm~\ref{alg} in linear time. The only modification
we need to make is to account for the coset codes
$\{\cC_j^{r,\bc_r}\}$, and this can be accomplished by replacing the
definition of $A_{k,j}$ and $B_{k,j}$ in \eqref{eq:B_kj_A_kj} by

\bra  A_{k,j} &\defined& \sum_{\bg\in \cC_j^{r,\bc_r}, \: g_k=1}
e^{-K\sum_{i\in \cN_j \setminus (\cN_r \cup \{k\})} u_{i,j} g_i} \nonumber \\
B_{k,j} &\defined& \sum_{\bg\in \cC_j^{r,\bc_r}, \: g_k=0}
e^{-K\sum_{i\in \cN_j \setminus (\cN_r \cup \{k\})} u_{i,j} g_i}
\label{eq:B_kj_A_kjreduced} \era These can be computed efficiently
for any coset code $\cC_j^{r,\bc_r}$ using the method described in
Section~\ref{sec: Akj and Bkj}. However, instead of
\eqref{eq:B_k_A_k} we now have (note that for simplicity we omit the
dependence on $j$) \bra \label{eq: B_k_A_k new} A_{k} &\defined&
\sum_{\bg \: : \: H^r \bg = \tilde{\bh}^{r,\bc_r}, \: g_k=1}
e^{-K \left( \sum_{i=0}^{n-1} u_{i} g_i - u_k \right)} \nonumber \\
B_{k} &\defined& \sum_{\bg \: : \: H^r \bg = \tilde{\bh}^{r,\bc_r},
\: g_k=0} e^{-K \sum_{i=0}^{n-1} u_{i} g_i} \: . \era It can be seen
that the only alteration we need to make is to replace \eqref{eq:
calcAk} and \eqref{eq: calcBk} by \bre A_k = e^{K u_k} \sum_{\bs}
E(k,\bs,1) \tilde{E}(k,\bs \oplus \tilde{\bh}^{r,\bc_r}) \label{eq:
calcAk coset}\ere and \bre B_k = \sum_{\bs} E(k,\bs,0)
\tilde{E}(k,\bs \oplus \tilde{\bh}^{r,\bc_r}) \label{eq: calcBk
coset} \ere for $k=0,\ldots,n-1$. For the calculation of $\D(\bu)$
of a coset code, we can appropriately modify and use the efficient
method from Section~\ref{sec: calculation of D(u)}. The necessary
modification for a coset code characterized by a vector
$\tilde{\bh}^{r,\bc_r}$, is that we need to recursively calculate
\bre \label{eq: coset code D(u)} A(n-1,\tilde{\bh}^{r,\bc_r}) \ere
rather than $A(n-1,0)$.

Solving Problem-$\text{P}^{r,\bc_r}$ thus finds an approximate
minimum-weight vector $\bhc$ subject to enforcing a local codeword on parity check $r$. 
We repeatedly apply this method, solving
Problem-$\text{P}^{r,\bc_r}$ over all nonzero local codewords
$\bc_r\in \cC_r \backslash \{\bzero_r\}$, and for each local
codeword we record both a lower bound
$\underline{l^r_{\bc_r}}=\D(\bu)$ and an upper bound
$\overline{l^r_{\bc_r}}=\P(\btlambda)$ (see \eqref{eq: sandwich}) on
the optimal value. In this manner, we create two sets of values
$\{\underline{l^r_{\bc_r}}\}_{\bc_r \in \cC_r \backslash
\{\bzero_r\} }$ and $\{\overline{l^r_{\bc_r}}\}_{\bc_r \in \cC_r
\backslash \{\bzero_r\} }$. Next, we calculate the minima \bre
\label{eq: l^rmin underline} \underline{l^r_{\text{min}}}
\defined \min_{\bc_r \in \cC_r \backslash \{\bzero_r\}} \underline{l^r_{\bc_r}} \ere
and \bre \label{eq: l^rmin overline} \overline{l^r_{\text{min}}}
\defined \min_{\bc_r \in \cC_r \backslash \{\bzero_r\}} \overline{l^r_{\bc_r}} \ere
Finally, the entire procedure is repeated for all $r \in \cJ$ and we
calculate \bre \label{eq: lmin} \underline{l_{\min}}\triangleq
\min_{r \in \cJ} \{\underline{l^r_{\text{min}}}\} \;, \quad
\overline{l_{\min}}\triangleq \min_{r \in \cJ}
\{\overline{l^r_{\text{min}}}\}\ere We refer to this process as {\em
Algorithm~\ref{alg lower bound}}, which is summarized in what
follows.

\begin{algorithm}
\label{alg lower bound} Given a GLDPC code, do:
\begin{itemize}
\item
{\bf loop over $r \in \cJ$}
\begin{itemize}
\item
{\bf loop over $\bc_r \in \cC_r \backslash \{\bzero_r\}$}
\begin{itemize}
\item
{\bf loop over $i \in \cN_r$ and $j \in \cN_i$}
\begin{itemize}
\item Compute $\tilde{\bh}_j^{r,\bc_r}$ using \eqref{eq: h_r,S0}
\item Define $H^r_j = \{\bh_i\}_{i \in \cN_j \backslash
\cN_r}$
\item Compute all codewords of the local coset code
$\cC_j^{r,\bc_r}$ defined by $H^r_j \tilde{\bc}_r =
\tilde{\bh}_j^{r,\bc_r}$
\end{itemize}
\item Run Algorithm~\ref{alg} to solve
Problem-$\text{P}^{r,\bc_r}$, defined by \eqref{eq:ml_lp}-\eqref{eq:
problem p_rs0,4}. While running the algorithm, use the definitions
of $A_{k,j}$ and $B_{k,j}$ in \eqref{eq:B_kj_A_kjreduced} instead of
those in \eqref{eq:B_kj_A_kj}
\item Set $\underline{l^r_{\bc_r}}=\D(\bu)$ and $\overline{l^r_{\bc_r}}=\P(\btlambda)$ where the
vectors $\bu$ and $\btlambda$ are output by Algorithm~\ref{alg}.
Calculate $\D(\bu)$ using \eqref{eq:lp_dual} or by the more
efficient method in Section~\ref{sec: calculation of D(u)}, with
\eqref{eq: coset code D(u)} instead of \eqref{eq: D(u) efficient}.
\end{itemize}
    \item Calculate $\underline{l^r_{\text{min}}}=\min_{\bc_r \in \cC_r \backslash \{\bzero_r\}}
    \underline{l^r_{\bc_r}}$ and $\overline{l^r_{\text{min}}}=\min_{\bc_r \in \cC_r \backslash \{\bzero_r\}}
    \overline{l^r_{\bc_r}}$
\end{itemize}
\item Output $\underline{l_{\min}}=\min_{r
\in \cJ} \{\underline{l^r_{\text{min}}}\}$ and
$\overline{l_{\min}}=\min_{r \in \cJ}
\{\overline{l^r_{\text{min}}}\}$
\end{itemize}
\end{algorithm}

\begin{proposition}
$\underline{l_{\min}}$, output by Algorithm~\ref{alg lower bound},
is a lower bound on the minimum distance of the code.
\end{proposition}
\begin{proof}
First, we claim that $\underline{l^r_{\text{min}}}$, defined in
\eqref{eq: l^rmin underline}, is a lower bound on the minimum
distance of the code, subject to the restriction that the local
codeword on node $r$ is nonzero. This is due to the fact that a
minimum-distance codeword $\tilde{\bc}^r$ subject to this
restriction must have some nonzero local codeword $\bc_r$ on check
node $r$, and while running Algorithm~\ref{alg} with the values of
$\bc_r$ forced on variable nodes $\cN_r$, we optimize over a set
containing $\tilde{\bc}^r$. Furthermore, due to weak duality, the
result $\underline{l^r_{\bc_r}}=\D(\bu)$ is a lower bound on the
optimal value. Consequently, $\underline{l^r_{\text{min}}}$ is a
lower bound on the minimum distance subject to the above
restriction. Now, since a (global) minimum-distance codeword
$\tilde{\bc}$ must contain a check node $r'$, the neighbors of which
represent a nonzero local codeword $\bc_{r'}$, then at some point
Algorithm~\ref{alg lower bound} will have passed through node $r'$
and the local codeword $\bc_{r'}$ and produced a lower bound on the
weight of $\tilde{\bc}$. Since we take the minimum value in
\eqref{eq: l^rmin underline} and \eqref{eq: lmin},
$\underline{l_{\min}}$ is a lower bound on the minimum distance of
the code.
\end{proof}
Recall that we assume parity check degrees which do not depend on
$N$. Consequently, for any check node $r$, the bounds
$\underline{l^r_{\text{min}}}$ and $\overline{l^r_{\text{min}}}$ are
obtained by running Algorithm~\ref{alg} a constant number of times.
Repeating this process over all check nodes thus entails a total
complexity of $O(N^2)$.

Next, we claim that $\overline{l_{\min}}$ is at least as large as
the fractional distance.
\begin{proposition}\label{prop: min-dist-bnd is upper bound on dfrac} The bound $\overline{l_{\min}}$ output by
Algorithm~\ref{alg lower bound} is an upper bound on the fractional
distance of the code.
\end{proposition}
\begin{proof}
Recall that $\overline{l_{\min}}$ was obtained by running
Algorithm~\ref{alg} over some reduced version of the code,
corresponding to the selection of a nonzero local codeword $\bc_r$
on some check node $r$. That is, we find the minimum of $\sum_{i \in
\cI} c_i$ over the original LP polytope $\cQ(\cC)$ while also
imposing the values $\bc_r$ of the local codeword. For $\bc_r =
\{c_{r,i}\}_{i \in \cN_r}$, let $\cQ_{\bc_r}(\cC)= \{
(p_1,\dots,p_N) \in \cQ(\cC): \forall i \in \cN_r, p_i=c_{r,i} \}$
be the polytope $\cQ(\cC)$ with the added restriction of the local
codeword $\bc_r$. Now, the minimum of any LP is attained at a
vertex. Consequently, if we show that $\cV(\cQ_{\bc_r}(\cC))
\subseteq \cV(\cQ(\cC))$, where $\cV(\cP)$ is the vertex set of the
polytope $\cP$, then the value of $\overline{l_{\min}}$ output by
Algorithm~\ref{alg lower bound} is an upper bound on the $L_1$-norm
of a nonzero pseudocodeword in $\cC$, and is thus an upper bound on
the fractional distance. To this end, it will suffice to show that
imposing an integer value on a single coordinate does not create new
vertices. Without loss of generality, let $\cQ^N(\cC)\triangleq
\cQ(\cC)\bigcap \{c_N =1\}$, and $\bv \in \cV(\cQ^N(\cC))$. We now
show that $\bv \in \cV(\cQ(\cC))$. In essence, an easy algebraic
argument will be used. First, note that $\cQ(\cC)$ can be
represented as a matrix inequality

\bre A \bc \succeq \bb \label{eq: Q A notation} \ere where $\bc$ and
$\bb$ are column vectors, $A$ is a matrix representing, together
with $\bb$, the constraints~\eqref{eq:omegaj_def2}\footnote{The
constraints~\eqref{eq:omegaj_def2}, which are valid for plain LDPC
codes, are clearly expressible in matrix form. In
Section~\ref{sec:GLDPC_polytope} we will see that for GLDPC codes,
the fundamental polytope is also expressible in the form \eqref{eq:
Q A notation}, and thus the conclusion of Proposition~\ref{prop:
min-dist-bnd is upper bound on dfrac} holds also for GLDPC codes.}
for all $j \in \cJ$, and $\succeq$ is the standard coordinate-wise
inequality. The polytope $\cQ^N(\cC)$ can be similarly represented
by

\bre A' \bc \succeq \bb' \nonumber \ere where

\bre A' = \left[\begin{array}{c} A \\ \hline 0 \; 0 \; \dots \; 0\;
1\end{array} \right], \quad \bb' = \left[\begin{array}{c} \bb
\\ \hline 1\end{array} \right] \nonumber \ere We make use of the property that
a vertex is the intersection of $N$ independent hyperplanes.
Specifically, if $\bv \in \cQ^N(\cC)$ then $\bv \in \cV(\cQ^N(\cC))$
if and only if there exists a subset of $N$ rows $A'_1$ of $A'$, and
a corresponding subvector $\bb'_1$ of $\bb'$ such that

\bre A'_1 \bc = \bb'_1 \nonumber \ere and $\det(A'_1) \neq 0$ (see,
e.g., \cite[p.~185]{bertsekas_convex}). Suppose $\bv \in
\cV(\cQ^N(\cC))$. If the matrix $A'_1$ does not contain the last row
of $A'$, then it is also a sub-matrix of $A$ and we conclude that
$\bv \in \cV(\cQ(\cC))$, as required. If $A'_1$ contains the last
row of $A'$, then we replace this row in $A'_1$ with

\bre [0\; 0\; \dots 0\; -1] \label{eq:minus1 row} \ere which appears
in $A$ due to the constraint $c_N \leq 1$ (see
\eqref{eq:omegaj_def2}); with this replacement the determinant is
still nonzero and we conclude that $\bv \in \cQ(\cC)$ is the
intersection of $N$ independent hyperplanes of $\cQ(\cC)$, and thus
$\bv \in \cV(\cQ(\cC))$.
\end{proof}

\emph{Discussion.} In this section it was shown that
$\underline{l_{\text{min}}}$ and $\overline{l_{\text{min}}}$, output
by Algorithm~\ref{alg lower bound} and obtained with complexity
$O(N^2)$, constitute a lower bound on the overall minimum distance
and an upper bound on the fractional distance, respectively. Due to
\eqref{eq: duality gap} and the proof of \cite[Theorem
1]{lpon_journal}, it can be seen that
$\overline{l_{\text{min}}}-\underline{l_{\text{min}}}\le \delta N$
where $\delta>0$ can be made as small as desired. Moreover, it is
possible to improve the lower bound on the minimum distance in the
following two ways. First, the node merging technique presented in
Section~\ref{sec:cuts for_improved_LPD} can be used to produce a
tighter LP relaxation, and thus a better lower bound on the minimum
distance. This approach is exemplified in Section~\ref{sec:
results}. Second, instead of examining single constraint nodes on
which a local codeword is forced, one could divide $\cJ$ into pairs
of constraint nodes, and for each pair examine all nonzero
combinations of their local codewords. For each such nonzero
combination of pairs of codewords, it is possible to find a
minimum-weight vector subject to forcing the variable node neighbors
to values corresponding to these codewords, similar to what is done
in Algorithm~\ref{alg lower bound}. As another possible variant to
Algorithm~\ref{alg lower bound}, we propose the following greedy
procedure, which we term the \emph{bar method}.

We set a target level or ``bar'', denoted $\text{BAR}$, for the
lower bound which the procedure attempts to exceed, as follows. We
loop over all check nodes, as in Algorithm~\ref{alg lower bound}.
For each check node $r\in \cJ$ we evaluate
$\underline{l^r_{\text{min}}}$, defined in \eqref{eq: l^rmin
underline}. If $\underline{l^r_{\text{min}}} \ge \text{BAR}$, we
proceed to the next check node. If $\underline{l^r_{\text{min}}} <
\text{BAR}$, say due to a local codeword $\bc^0_r$ for which
$\underline{l_{\bc^0_r}}<\text{BAR}$, we refine the bound
$\underline{l_{\bc^0_r}}$ as follows. First we pick a check node
$j\in \cN_i$ for some $i \in \cN_r$; thus $j$ is a check node at
distance $2$ from $r$, with respect to $\cG$. Next we loop over the
local codewords $\{\bc_j\}$ on node $j$ which coincide with
$\bc^0_r$ (i.e., which have the same values on the variable nodes in
$\cN_j \cap \cN_r$). For each such local codeword, we run
Algorithm~\ref{alg}, forcing local codewords on \emph{both} check
nodes $r$ and $j$. This forcing of local codewords jointly on a pair
of check nodes produces a lower bound
$\underline{l_{\bc^0_r,\bc_j}}$ on the minimum distance subject to
forcing both check nodes to these local codewords. Define
$\underline{l^{\bc^0_r,j}_{\min}}\defined\min_{\bc_j \in
\cC_j}\{\underline{l_{\bc^0_r,\bc_j}}: \quad \bc_j \; \text{
coincides with } \bc^0_r\}$. If $\underline{l^{\bc^0_r,j}_{\min}}\ge
\text{BAR}$, we conclude that the bar has been exceeded for check
node $r$, and proceed to the next check node. Otherwise, we pick
another check node $j\in \cN_i$ for some $i \in \cN_r$ and repeat
the process until for some choice of node $j$ at distance $2$ from
$r$ the bar is exceeded. If the bar is not exceeded for all nodes
$j$ at distance $2$ from $r$, the procedure is terminated with
failure. If, on the other hand, for all $r\in \cJ$ we have
$\underline{l^r_{\text{min}}} \ge \text{BAR}$ or otherwise if for
all $\bc^0_r\in \cC_r$ we have $\underline{l^{\bc^0_r,j}_{\min}}\ge
\text{BAR}$ for some node $j$ at distance $2$ from $r$, then we say
that the level $\text{BAR}$ is \emph{attainable}. Due to the bound
refining process, an attainable bar level is guaranteed to be a
lower bound on the minimum distance. Using initial lower and upper
values for the bar (which can be determined ad-hoc), we can run the
procedure, each time using a different level for the bar, and use
the bisection method to find the highest attainable bar level.
\section{A Lower Bound on the Fractional Distance in Quadratic
Complexity} \label{sec: lower bound on dfrac} In \cite{lpdecode},
the following algorithm for calculating the fractional distance
$d_{\text{frac}}$ was proposed. Consider the codeword polytope
$\cQ(\cC)$ and the all-zero vertex $\bzero$. It was shown in
\cite{lpdecode} that the facets of this polytope (for plain LDPC
codes) are given by \eqref{eq:omegaj_def2}, $\forall j \in \cJ$.
Denote the set of facets of $\cQ(\cC)$ which do not contain $\bzero$
by $\cF$. For each facet $f \in \cF$, run an LP solver to find the
minimum $L_1$ norm $\sum_{i\in \cI} c_i$ over $f$. The smallest
value (over all facets) obtained in this procedure is the fractional
distance \cite{lpdecode}. The complexity of this calculation is the
same as the complexity of running an LP solver $|\cF|$ times, and in
our case we have $|\cF| = O(N)$. By using the iterative
linear-complexity LP decoder, we will demonstrate a procedure which
produces a lower bound on the fractional distance with complexity
$O(N^2)$. In terms of computational complexity, this compares
favorably with the aforementioned procedure because an LP solver in
general has complexity worse than $O(N)$. Furthermore, if a higher
complexity is allowed, this lower bound can be made arbitrarily
close to the true fractional distance. In this section we assume
plain LDPC codes. The results in the next section enable the
generalization of the algorithm to GLDPC codes.

Recall that $\bc \in \cQ(\cC)$ if and only if $0 \leq c_i \leq 1$
$\forall i \in \cI$ and $\forall j \in \cJ, S \subseteq \cN_j, |S|
\text{ odd}$ we have $$\sum_{i\in\cN_j\setminus S} c_i + \sum_{i\in
S} \left( 1-c_i \right) \geq 1$$ Of these inequality constraints,
the facets of $\cQ(\cC)$ which do not contain $\bzero$ are \bre
\{c_i = 1\}_{i \in \cI} \label{ci=1 facets} \ere and
\begin{multline} \label{diagonal facets} \sum_{i\in\cN_j\setminus S} c_i + \sum_{i\in S}
\left( 1-c_i \right) = 1
\\ \forall j \in \cJ, S \subseteq \cN_j, |S|
\text{ odd}, |S|>1 \end{multline} We follow the approach in
\cite{lpdecode}, i.e., for each of the facets in \eqref{ci=1
facets}-\eqref{diagonal facets} we evaluate the minimum $L_1$ norm
over the facet. We show how each such evaluation can be performed in
linear time using Algorithm~\ref{alg}.

Let $i \in \cI$ be some index and consider its corresponding facet
from \eqref{ci=1 facets}. To find the minimum weight over this
facet, we can implement the method from Section~\ref{sec:min
distance bound}, as follows. The Tanner graph is modified by
removing variable node $i$ and all the edges incident to it. Now, to
keep the graph consistent with the original code, each constraint
node $j \in \cN_i$ (be it a standard parity check or generalized
node) is modified so that it represents a local coset code $H_{j
\setminus i} \bc_{j \setminus i}= \bh_{i,j}$ (where the subscript
$j\setminus i$ denotes that we remove the $i$'th column of $H_j$ and
the $i$'th bit of $\bc_j$) for some column vector $\bh_{i,j}\neq 0$
rather than $H_j \bc_j= 0$. The column $\bh_{i,j}$ is the column
from the parity-check matrix representing constraint node $j$, which
corresponds to the index of variable node $i \in \cN_j$ (note that
if $j$ is a standard parity check then $\bh_{i,j}=1$ has dimesion
$1$). Algorithm~\ref{alg} runs the same on this modified graph,
except that, similarly to Section~\ref{sec:min distance bound}, we
need to replace the calculation of $A_{k,j}$ and $B_{k,j}$ in
\eqref{eq: calcAk} and \eqref{eq: calcBk} by \eqref{eq: calcAk
coset} and \eqref{eq: calcBk coset}, respectively. The vector
$\bh_{i,j}$ replaces $\tilde{\bh}^{r,\bc_r}$ in \eqref{eq: calcAk
coset} and \eqref{eq: calcBk coset}. If we take the value of the
dual in the execution of Algorithm~\ref{alg} then due to \eqref{eq:
sandwich} we get for each $i\in \cI$, a vector with weight
$\underline{d_{\text{frac},i}^{(1)}}$ which is a lower bound on the
minimum weight over the facet $\{c_i=1\}$. The minimum value, \bre
\label{eq: bound on dfrac,square facets} d_{\text{frac}}^{(1)}
\defined \min_{i \in \cI} \underline{d_{\text{frac},i}^{(1)}}\ere is thus a lower bound on the minimum fractional weight over all facets
\eqref{ci=1 facets} of $\cQ(\cC)$. This lower bound can be made as
tight as desired if in Algorithm~\ref{alg} we take $K$ large enough
and $\epsilon_0$ small enough.

We now turn to the problem of calculating minimum fractional
distance on the facets \eqref{diagonal facets}. Let $r$ be some
fixed constraint node and let $S\subseteq \cN_r$ be an odd-sized
set, $|S|>1$. Define the hyperplane $R^{r,S}$ as follows:
$$
R^{r,S}= \left\{\bc:\sum_{i\in\cN_r\setminus S} c_i + \sum_{i\in S}
\left( 1-c_i \right) = 1 \right\}
$$
Our problem is to find the minimum fractional distance on
$R^{r,S}\cap \cQ(\cC)$, denoted $d_{\text{frac}}^{r,S}$. Consider
the following problem, {\em Problem-}$P_{\text{frac}}^{r,S,B}$: \bre
\min_{\bc,\bomega} \sum_{i\in \cI} c_i + B
\left(\sum_{i\in\cN_r\setminus S} c_i + \sum_{i\in S} \left( 1-c_i
\right) - 1 \right) \label{eq: problem Prs}\ere subject to
~\eqref{eq:positive_w}-\eqref{eq:ci}, where $B>0$ is some large
constant. This definition is motivated by the following
observations. First, the feasible region of
Problem-$\text{P}_{\text{frac}}^{r,S,B}$ (in the $\bc$ variables) is
by definition the polytope $\cQ(\cC)$. Thus, every feasible point
must satisfy \bre \sum_{i\in\cN_r\setminus S} c_i + \sum_{i\in S}
\left( 1-c_i \right) \geq 1 \label{eq: half hyperplane}\ere The
second term in \eqref{eq: problem Prs} is, by \eqref{eq: half
hyperplane}, a positive penalty term on the event of sharp
inequality in \eqref{eq: half hyperplane}. Therefore, in the limit
where $B \rightarrow \infty$, the exact value
$d_{\text{frac}}^{r,S}$ is produced as the solution to
Problem-$\text{P}_{\text{frac}}^{r,S,B}$. If the constant $B$ is
finite, the weight $d_{\text{frac},r,S,B}^{(2)}$ of the solution to
Problem-$\text{P}_{\text{frac}}^{r,S,B}$ can be seen to be a lower
bound on $d_{\text{frac}}^{r,S}$ (since the vector $\bc$ which
attains the minimum fractional distance $d_{\text{frac}}^{r,S}$ is
feasible in Problem-$\text{P}_{\text{frac}}^{r,S,B}$ and yields the
objective function value $d_{\text{frac}}^{r,S}$). Second, the
objective function of Problem-$\text{P}_{\text{frac}}^{r,S,B}$ is
linear in its variables (it also contains the additive constant
$B(|S|-1)$ which is independent of the variables and thus can be
ignored). Consequently, Problem-$\text{P}_{\text{frac}}^{r,S,B}$ can
be reduced to an instance of Problem-P, by setting $\{\gamma_i\}_{i
\in \cI}$ as \bre \tilde{\gamma}_i = \left\{\begin{array}{l l} 1 & i
\notin \cN_r
\\ 1-B & i \in S \\ 1+B & i \in \cN_r / S
\end{array} \right. \label{eq: modified gammas}\ere
We conclude that the solution to
Problem-$\text{P}_{\text{frac}}^{r,S,B}$ can be approximated
arbitrarily closely by Algorithm~\ref{alg}. Now, we use
Algorithm~\ref{alg} to solve
Problem-$\text{P}_{\text{frac}}^{r,S,B}$ $\forall j \in \cJ, S
\subseteq \cN_j, |S| \text{ odd}, |S|>1$. By taking the dual value
output by Algorithm~\ref{alg}, again by \eqref{eq: sandwich}, we
obtain a set of lower bounds
$\{\underline{d_{\text{frac},r,S,B}^{(2)}}\}$ on the minimum
fractional weight for each of the facets \eqref{diagonal facets} of
$\cQ(\cC)$. A lower bound on the minimum fractional distance over
the facets \eqref{diagonal facets} is given by
$d_{\text{frac},B}^{(2)}
\defined \min_{r \in \cJ, S \subseteq \cN_r, |S| \text{ odd}, |S|>1}
\underline{d_{\text{frac},r,S,B}^{(2)}}$ Finally, we combine this
result with \eqref{eq: bound on dfrac,square facets} and obtain the
lower bound
$$
d_{\text{frac}} \geq d_{\text{frac},B} \defined
\min(d_{\text{frac}}^{(1)},d_{\text{frac},B}^{(2)})
$$

Assuming the node degrees in the Tanner graph are bounded by a
constant independent of $N$, the overall number of calls to
Algorithm~\ref{alg} is $O(N)$. Thus the computational complexity of
evaluating a lower bound on the minimum fractional distance using
the procedure described above is $O(N^2)$.

Naturally, the exact fractional distance can be approached by using
very large values for the penalty constant $B$. It should be noted
that this could theoretically have an effect on the bound on the
convergence rate of Algorithm~\ref{alg}: In \cite[Theorem
1]{lpon_journal} it was shown that the bound on the convergence rate
of Algorithm~\ref{alg} is related to $\gamma_{\max}$, and from
\eqref{eq: modified gammas} it can be seen that this quantity may be
large as we increase $B$. In Section~\ref{sec: results}, we describe
several experiments conducted with $B=10N$. In these experiments, we
have not observed a significant increase in running time as compared
with the minimum distance bounds from Section~\ref{sec:min distance
bound}.

\section{Fundamental polytopes of GLDPC and nonbinary codes} \label{sec:GLDPC_polytope}
For plain LDPC codes, the fundamental polytope is represented by
\eqref{eq:omegaj_def2}. In this section, we propose a practical
procedure which obtains representations of the fundamental polytopes
of two important classes of codes: GLDPC and nonbinary codes. Using
these representations, one can apply a procedure similar to the one
presented in Section~\ref{sec: lower bound on dfrac} to calculate a
tight lower bound on the fractional distance which can be used to
assess the performance of the LP decoder in these cases. An
interesting effort for nonbinary codes has recently been made by
Skachek~\cite{Skachek2010lpd}, for trenary codes. We propose a
general practical technique which relies on the double description
method \cite{FukudaProdon1995,Motzkin} to find, similar to
\eqref{eq:omegaj_def2}, a description of the fundamental polytope.
For the case of nonbinary codes, as in \cite{Skachek2010lpd}, we
consider the LP formulation from \cite{flanagan2008lpd} which
expresses a nonbinary decoding problem using a binary problem of
higher dimension.

Consider a constraint node $j$. Suppose it represents some general,
not necessarily linear, local code $\cC_j$ with $M$ codewords and
block length $d$ (in what follows, we will omit the subscript $j$).
Similar to~\eqref{eq:tpsi_j_def} we construct a $d \times M$ matrix
$\Psi$ for this code by writing all the codewords on the columns of
that matrix, i.e.
$$
\Psi = \left(
  \begin{array}{cccc}
    \bc^{0} & \bc^{1} & \ldots & \bc^{M-1} \\
  \end{array}
\right)
$$
where $\bc^{0}, \bc^{1}, \ldots, \bc^{M-1}$ are the codewords. We
also define
$$
\tPsi = \left(
        \begin{array}{c}
          \bone_M \\
          \Psi \\
        \end{array}
      \right)
$$
where $\bone_M$ is a row vector of $M$ ones. The code polytope,
$\cQ(\cC)$, is defined by (see Section~\ref{sec:cuts
for_improved_LPD} for an equivalent definition)
$$
\cQ(\cC) \defined \left\{ \bc \: : \: \exists \bomega \succeq 0
\mbox{ such that } \tPsi \bomega = \onec \right\}
$$
Following~\cite[Theorems 4.9 and 4.10, pp.~97--98]{wolsey}, we now
show that $\cQ(\cC)$ is indeed a polytope, and present it in a more
explicit form. Consider the following two systems of linear
inequalities \bre \label{eq:alternative1} \bomega \succeq 0 \: , \:
\tPsi \bomega = \onec \ere and \bre \label{eq:alternative2} \tPsi^T
\bv \preceq 0 \: , \: (1 \: \: \: \bc^T) \bv > 0 \ere By Farkas'
lemma (e.g.~\cite[page 263]{boyd}) the
systems~\eqref{eq:alternative1} and~\eqref{eq:alternative2} are
strong alternatives. That is, the system~\eqref{eq:alternative1} is
feasible if and only if the system~\eqref{eq:alternative2} is
infeasible. We conclude that, \bre \label{eq:cpcc1} \cQ(\cC) =
\left\{ \bc \: : \: (1 \: \: \: \bc^T) \bv \le 0 \quad \forall \bv
\in \cT \right\} \ere where $\cT$ is the following polyhedral cone
\bre \label{eq:ct1} \cT = \left\{ \bv \: : \: \tPsi^T \bv \preceq 0
\right\} \ere Now, $\cT$ can be expressed in the following
alternative form, \bre \label{eq:ct2} \cT = \left\{ \bv \: : \: \bv
= \sum_{l=1}^r \mu_l \bv_l \: , \: \left\{ \mu_l \ge 0
\right\}_{l=1}^r \right\} \ere where $\bv_l$, $l=1,\ldots,r$ are the
extreme rays of $\cT$. By~\cite[Proposition 4.3, p.~94]{wolsey}
$\bv_l$ is an extreme ray of a polyhedral cone $\cT$ if and only if
$\{ \eta \bv_l \: : \: \eta \ge 0 \}$ is a one dimensional face of
$\cT$. For the polyhedron $\cT \subseteq \mathbb{R}^{d+1}$, the
vector $\bv\in\cT$, $\bv \ne \bzr$ is an extreme ray if and only if
it satisfies $d$ linearly independent constraints among $\tPsi^T \bv
\preceq 0$ with equality. By~\eqref{eq:cpcc1} and~\eqref{eq:ct2} we
have \bre \label{eq:cpcc2} \cQ(\cC) = \left\{ \bc \: : \: (1 \: \:
\: \bc^T) \bv_j \le 0 \quad \forall j=1,\ldots,r \right\} \ere Thus,
if we can obtain the extreme rays of the cone $\cT$, we have the
desired representation of the code polytope using \eqref{eq:cpcc2}.

One possible way of doing this is to apply the double description
method \cite{FukudaProdon1995, Motzkin}. A pair of matrices $(A,R)$
is said to be a double description (DD) pair if \bre A \bx \succeq 0
\quad \text{iff} \quad \bx = R \bmu, \; \bmu \succeq 0 \ere Clearly,
\eqref{eq:ct1} and \eqref{eq:ct2} lead to a DD pair of the cone
$\cT$, by identifying $A=-\tPsi^T$, $\bx=\bv$, $R=(\; \bv_1 \; \bv_2
 \dots \bv_r \;)$ and $\bmu = (\mu_1 \: \mu_2 \: \dots \:
\mu_r)^T$. The problem is, given one of the matrices $R$ or $A$, to
find the other. In our case, we need to find $R$ given $A$. The
double description method is an algorithm which solves this problem.
Unfortunately, the complexity of the algorithm in general grows
exponentially with the size of the problem, so it will only work
reasonably if we keep the code small. To apply the DD method, we
have used the publicly available software implementation cdd+
\cite{Fukuda2008}.

Using cdd+, we were able to obtain a representation of a binary
constraint node representing the $(7,4)$ Hamming code $\cH_7$. The
matrix $R_{\cH_7}$ has $8$ rows and $70$ columns, and its transpose
is given in Figure~\ref{fig: R^T of hamming}. For example, the first
line of $R^T_{\cH_7}$ represents the constraint
\begin{equation*}
-2 -c_1 +c_3+c_5+c_7 \le 0
\end{equation*}
which is a facet of $\cQ(\cH_7)$. Note that the constraints $0\le
c_i \le 1 \; \forall i\in \{1,2,\dots,7\}$ are also implicit in
$R_{\cH_7}$. Using this representation together with the lower bound
on the fractional distance in Section~\ref{sec: lower bound on
dfrac}, we found a randomly-generated GLDPC code with $70$ variable
nodes and $20$ constraint nodes, each representing a local $\cH_7$
code, with fractional distance at least $11.2418$. This code has
rate $1/7$. The fractional distance result guarantees that this code
can correct $5$ errors using the LP decoder.
\begin{figure}
\begin{equation*}  R^T_{\cH_7} = \left(
{\footnotesize
\begin{tabular}{c c c c c c c c}
 -2 & -1 & 0 & 1 & 0 & 1 & 0 & 1 \\
 -2 & 0 & -1 & 1 & 0 & 0 & 1 & 1 \\
 -2 & 0 & -1 & 1 & 1 & 1 & 0 & 0 \\
 -2 & -1 & 0 & 1 & 1 & 0 & 1 & 0 \\
 -2 & -1 & 1 & 0 & 0 & 1 & 1 & 0 \\
 -1 & 0 & 0 & 0 & 0 & 1 & 0 & 0 \\
 -2 & 0 & 0 & 0 & -1 & 1 & 1 & 1 \\
 -1 & 0 & 0 & 0 & 0 & 0 & 1 & 0 \\
 -2 & 0 & 0 & 0 & 1 & 1 & 1 & -1 \\
 -2 & 1 & -1 & 0 & 0 & 1 & 1 & 0 \\
 -2 & 0 & 1 & -1 & 0 & 0 & 1 & 1 \\
 -2 & 0 & 1 & -1 & 1 & 1 & 0 & 0 \\
 -2 & 1 & 0 & -1 & 0 & 1 & 0 & 1 \\
 -2 & 1 & 0 & -1 & 1 & 0 & 1 & 0 \\
 -2 & 1 & -1 & 0 & 1 & 0 & 0 & 1 \\
 -1 & 0 & 0 & 0 & 1 & 0 & 0 & 0 \\
 -2 & 0 & 0 & 0 & 1 & -1 & 1 & 1 \\
 -1 & 0 & 0 & 0 & 0 & 0 & 0 & 1 \\
 -2 & 0 & 0 & 0 & 1 & 1 & -1 & 1 \\
 -2 & -1 & 1 & 0 & 1 & 0 & 0 & 1 \\
 -2 & 0 & 1 & 1 & 0 & 0 & -1 & 1 \\
 -2 & 1 & 0 & 1 & 0 & -1 & 0 & 1 \\
 -2 & 1 & 0 & 1 & 1 & 0 & -1 & 0 \\
 -2 & 0 & 1 & 1 & 1 & -1 & 0 & 0 \\
 -2 & 1 & 1 & 0 & 0 & -1 & 1 & 0 \\
 -2 & 1 & 1 & 0 & 1 & 0 & 0 & -1 \\
 -1 & 1 & 0 & 0 & 0 & 0 & 0 & 0 \\
 -2 & 1 & 1 & 0 & -1 & 0 & 0 & 1 \\
 -2 & 1 & 1 & 0 & 0 & 1 & -1 & 0 \\
 -1 & 0 & 1 & 0 & 0 & 0 & 0 & 0 \\
 -2 & 0 & 1 & 1 & -1 & 1 & 0 & 0 \\
 -2 & 1 & 0 & 1 & -1 & 0 & 1 & 0 \\
 -2 & 1 & 0 & 1 & 0 & 1 & 0 & -1 \\
 -2 & 0 & 1 & 1 & 0 & 0 & 1 & -1 \\
 -1 & 0 & 0 & 1 & 0 & 0 & 0 & 0 \\
 0 & -1 & 0 & 1 & -1 & 0 & -1 & 0 \\
 0 & 0 & -1 & 1 & -1 & -1 & 0 & 0 \\
 0 & 0 & -1 & 1 & 0 & 0 & -1 & -1 \\
 0 & -1 & 0 & 1 & 0 & -1 & 0 & -1 \\
 0 & -1 & 1 & 0 & -1 & 0 & 0 & -1 \\
 0 & 0 & 0 & 0 & -1 & 1 & -1 & -1 \\
 0 & 0 & 0 & 0 & -1 & 0 & 0 & 0 \\
 0 & 0 & 0 & 0 & -1 & -1 & 1 & -1 \\
 0 & 0 & 0 & 0 & 0 & 0 & 0 & -1 \\
 0 & 1 & -1 & 0 & -1 & 0 & 0 & -1 \\
 0 & 0 & 1 & -1 & -1 & -1 & 0 & 0 \\
 0 & 0 & 1 & -1 & 0 & 0 & -1 & -1 \\
 0 & 1 & 0 & -1 & -1 & 0 & -1 & 0 \\
 0 & 1 & 0 & -1 & 0 & -1 & 0 & -1 \\
 0 & 1 & -1 & 0 & 0 & -1 & -1 & 0 \\
 0 & 0 & 0 & 0 & 1 & -1 & -1 & -1 \\
 0 & 0 & 0 & 0 & 0 & -1 & 0 & 0 \\
 0 & 0 & 0 & 0 & -1 & -1 & -1 & 1 \\
 0 & 0 & 0 & 0 & 0 & 0 & -1 & 0 \\
 0 & -1 & 1 & 0 & 0 & -1 & -1 & 0 \\
 0 & 0 & 0 & -1 & 0 & 0 & 0 & 0 \\
 0 & 0 & -1 & -1 & 1 & -1 & 0 & 0 \\
 0 & 0 & -1 & -1 & 0 & 0 & -1 & 1 \\
 0 & -1 & 0 & -1 & 1 & 0 & -1 & 0 \\
 0 & -1 & 0 & -1 & 0 & -1 & 0 & 1 \\
 0 & 0 & -1 & -1 & 0 & 0 & 1 & -1 \\
 0 & 0 & -1 & -1 & -1 & 1 & 0 & 0 \\
 0 & -1 & 0 & -1 & 0 & 1 & 0 & -1 \\
 0 & -1 & 0 & -1 & -1 & 0 & 1 & 0 \\
 0 & 0 & -1 & 0 & 0 & 0 & 0 & 0 \\
 0 & -1 & -1 & 0 & 1 & 0 & 0 & -1 \\
 0 & -1 & -1 & 0 & 0 & -1 & 1 & 0 \\
 0 & -1 & -1 & 0 & -1 & 0 & 0 & 1 \\
 0 & -1 & -1 & 0 & 0 & 1 & -1 & 0 \\
 0 & -1 & 0 & 0 & 0 & 0 & 0 & 0
\end{tabular} }\right)
\end{equation*} \caption{The extreme rays of the cone
$\cT$ of the $(7,4)$ Hamming code.} \label{fig: R^T of hamming}
\end{figure}

\begin{figure*}[!t]
\begin{equation*}  R^T_{\text{SPC}_{GF(4)}(4)} = \left(
{\small
\begin{tabular}{c c c c c c c c c c c c c}
-1 & 0 & 0 & 0 & 0 & 0 & 0 & 0 & 0 & 0 & 1 & 1 & 1 \\
 -1 & 0 & 0 & 0 & 0 & 0 & 0 & 1 & 1 & 1 & 0 & 0 & 0 \\
 -2 & 0 & 1 & 1 & 0 & -1 & -1 & 0 & 1 & 1 & 0 & 1 & 1 \\
 -2 & 0 & -1 & -1 & 0 & 1 & 1 & 0 & 1 & 1 & 0 & 1 & 1 \\
 -2 & 1 & 0 & 1 & -1 & 0 & -1 & 1 & 0 & 1 & 1 & 0 & 1 \\
 -2 & -1 & 0 & -1 & 1 & 0 & 1 & 1 & 0 & 1 & 1 & 0 & 1 \\
 -2 & 0 & 1 & 1 & 0 & 1 & 1 & 0 & 1 & 1 & 0 & -1 & -1 \\
 -2 & 1 & 0 & 1 & 1 & 0 & 1 & 1 & 0 & 1 & -1 & 0 & -1 \\
 -1 & 1 & 1 & 1 & 0 & 0 & 0 & 0 & 0 & 0 & 0 & 0 & 0 \\
 -2 & 1 & 1 & 0 & 1 & 1 & 0 & 1 & 1 & 0 & -1 & -1 & 0 \\
 -2 & 1 & 1 & 0 & 1 & 1 & 0 & -1 & -1 & 0 & 1 & 1 & 0 \\
 -2 & 1 & 0 & 1 & 1 & 0 & 1 & -1 & 0 & -1 & 1 & 0 & 1 \\
 -2 & 0 & 1 & 1 & 0 & 1 & 1 & 0 & -1 & -1 & 0 & 1 & 1 \\
 -1 & 0 & 0 & 0 & 1 & 1 & 1 & 0 & 0 & 0 & 0 & 0 & 0 \\
 -2 & 1 & 1 & 0 & -1 & -1 & 0 & 1 & 1 & 0 & 1 & 1 & 0 \\
 -2 & -1 & -1 & 0 & 1 & 1 & 0 & 1 & 1 & 0 & 1 & 1 & 0 \\
 0 & 0 & 0 & 0 & 0 & 0 & 0 & 0 & 0 & -1 & 0 & 0 & 0 \\
 0 & 0 & 0 & 0 & 0 & 0 & 0 & 0 & 0 & 0 & 0 & 0 & -1 \\
 0 & 1 & 0 & 1 & -1 & 0 & -1 & -1 & 0 & -1 & -1 & 0 & -1 \\
 0 & -1 & 0 & -1 & 1 & 0 & 1 & -1 & 0 & -1 & -1 & 0 & -1 \\
 0 & 0 & 1 & 1 & 0 & -1 & -1 & 0 & -1 & -1 & 0 & -1 & -1 \\
 0 & 0 & -1 & -1 & 0 & 1 & 1 & 0 & -1 & -1 & 0 & -1 & -1 \\
 0 & 0 & 0 & 0 & 0 & 0 & 0 & 0 & 0 & 0 & 0 & -1 & 0 \\
 0 & 0 & 0 & 0 & 0 & 0 & 0 & 0 & -1 & 0 & 0 & 0 & 0 \\
 0 & 1 & 1 & 0 & -1 & -1 & 0 & -1 & -1 & 0 & -1 & -1 & 0 \\
 0 & -1 & -1 & 0 & 1 & 1 & 0 & -1 & -1 & 0 & -1 & -1 & 0 \\
 0 & 0 & 0 & 0 & 0 & 0 & 0 & -1 & 0 & 0 & 0 & 0 & 0 \\
 0 & 0 & 0 & 0 & 0 & 0 & 0 & 0 & 0 & 0 & -1 & 0 & 0 \\
 0 & 0 & -1 & -1 & 0 & -1 & -1 & 0 & 1 & 1 & 0 & -1 & -1 \\
 0 & -1 & 0 & -1 & -1 & 0 & -1 & 1 & 0 & 1 & -1 & 0 & -1 \\
 0 & 0 & 0 & 0 & 0 & 0 & -1 & 0 & 0 & 0 & 0 & 0 & 0 \\
 0 & -1 & 0 & -1 & -1 & 0 & -1 & -1 & 0 & -1 & 1 & 0 & 1 \\
 0 & 0 & -1 & -1 & 0 & -1 & -1 & 0 & -1 & -1 & 0 & 1 & 1 \\
 0 & 0 & 0 & 0 & 0 & -1 & 0 & 0 & 0 & 0 & 0 & 0 & 0 \\
 0 & -1 & -1 & 0 & -1 & -1 & 0 & 1 & 1 & 0 & -1 & -1 & 0 \\
 0 & -1 & -1 & 0 & -1 & -1 & 0 & -1 & -1 & 0 & 1 & 1 & 0 \\
 0 & 0 & 0 & 0 & -1 & 0 & 0 & 0 & 0 & 0 & 0 & 0 & 0 \\
 0 & 0 & 0 & -1 & 0 & 0 & 0 & 0 & 0 & 0 & 0 & 0 & 0 \\
 0 & 0 & -1 & 0 & 0 & 0 & 0 & 0 & 0 & 0 & 0 & 0 & 0 \\
 0 & -1 & 0 & 0 & 0 & 0 & 0 & 0 & 0 & 0 & 0 & 0 & 0
\end{tabular} }\right)
\end{equation*} \caption{The extreme rays of the cone
$\cT$ of the SPC($4$) code over GF($4$).} \label{fig: R^T of SPC4}
\end{figure*}

Flanagan \etal \cite{flanagan2008lpd} gave an LP formulation
suitable for nonbinary codes over rings. For a ring
$\cR=\{0,a_1,\dots,a_{|\cR|-1}\}$ of size $|\cR|$, every nonbinary
symbol is represented using a binary vector of size $|\cR|-1$. If
the nonbinary symbol is zero in some codeword, the corresponding
binary vector is set to zero. If the nonbinary symbol is nonzero,
say $a_i$, then the $i$'th element of the binary vector is set to
$1$ and the rest of the binary vector is set to zero. In this
manner, a nonbinary local codeword $\bc_j$ is mapped to a binary
local codeword. The convex hull of all the resulting binary
codewords corresponding to constraint node $j$ yields the
fundamental polytope $\cQ(\cC_j)$ and the overall polytope is the
intersection $\cap_{j\in\cJ}\cQ(\cC_j)$, just as in the binary case.
The difference is that in the nonbinary case, the dimension of the
polytope is larger, and because of the method of representation,
$\cC_j$, when viewed as a binary code, is not necessarily linear.

Clearly, using the DD method, it is possible to find
(experimentally) the representation of any local code, provided it
is not too large. As an example, using a binary vector of length
$12$, we have found a representation of the fundamental polytope of
the nonbinary simple parity check code of length $4$ over GF($4$).
This polytope has $40$ facets, given by the rows of the matrix
$R^T_{\text{SPC}_{GF(4)}(4)}$, which is presented in
Figure~\ref{fig: R^T of SPC4}. For any particular LDPC code over
GF($4$) which uses $\text{SPC}_{GF(4)}(4)$ as the local code in its
constraint nodes, we can use the method in Section~\ref{sec: lower
bound on dfrac} to obtain tight lower bounds on the fractional
distance. The same conclusion applies in general when, instead of
$\text{SPC}_{GF(4)}(4)$ we take any other code, provided that the DD
method outputs (in reasonable time complexity) a representation of
its corresponding fundamental polytope. We expect that this will be
the case for practical constituent codes.


\begin{figure*}[!t]
\begin{center}
\includegraphics[scale=0.7]{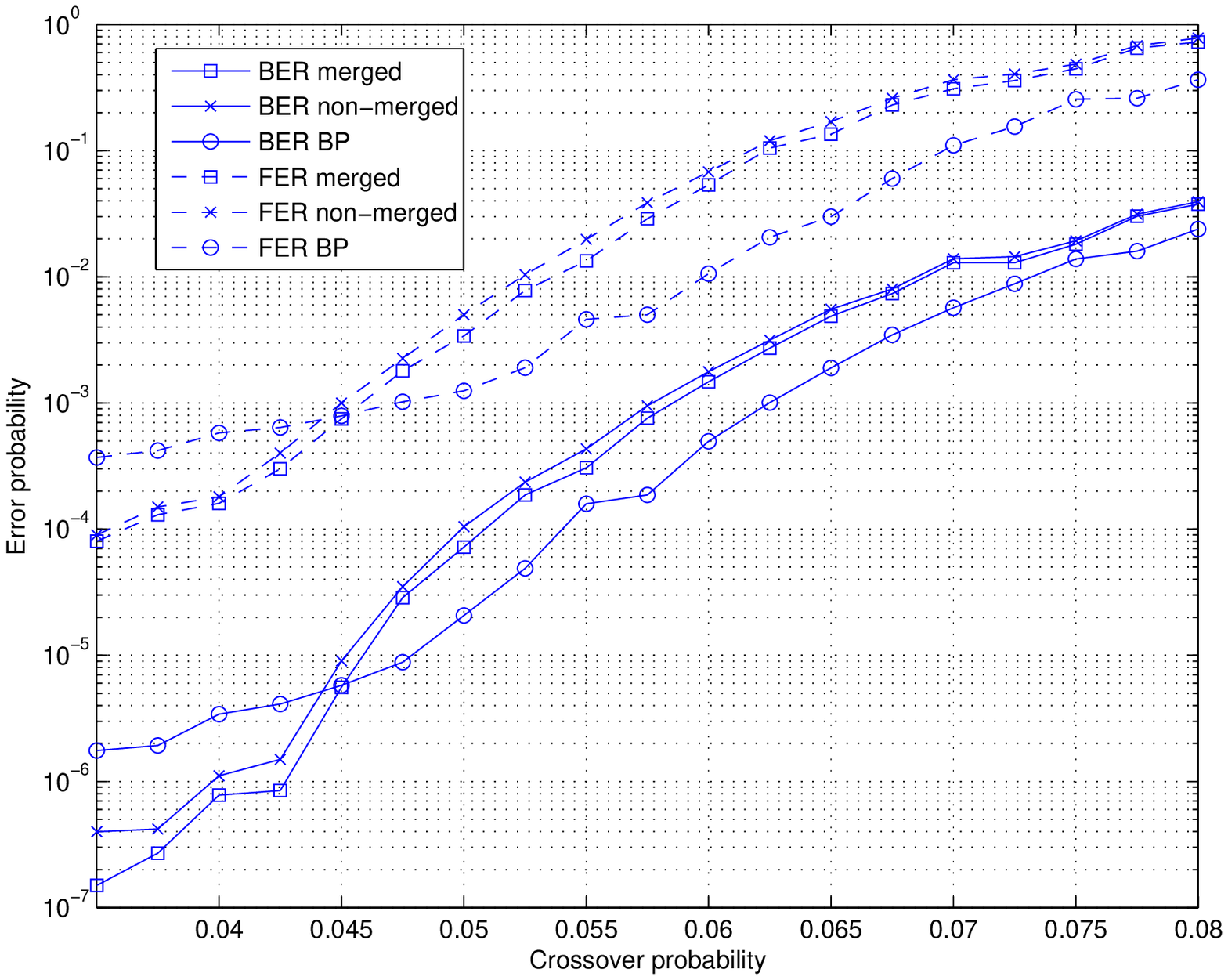}
\caption{Simulation results for the BSC comparing the approximate LP
decoder with and without merging of check nodes with belief
propagation.} \label{fig: simulation results}
\end{center}
\end{figure*}
\begin{figure*}[!t]
\begin{center}
\includegraphics[scale=0.7]{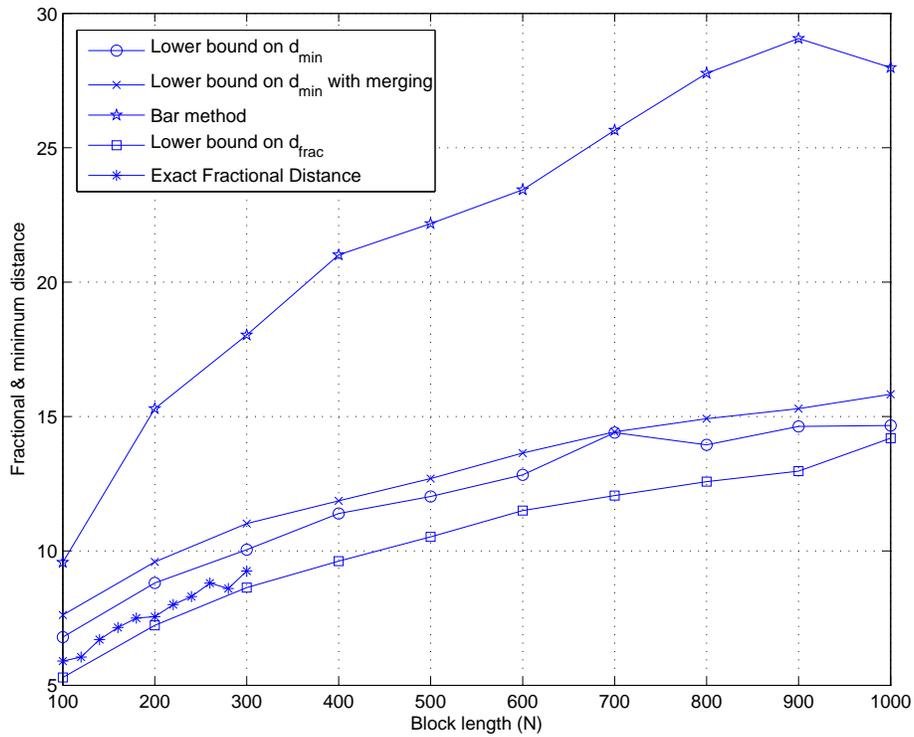}
\caption{Lower bounds on the fractional distance and minimum
distance.} \label{fig: mindistbarfracdist}
\end{center}
\end{figure*}

\section{Numerical Results}\label{sec: results}
In the experiments outlined below, Algorithm~\ref{alg} was operated
in the following mode, with respect to the values of the parameters
$K$ and $\epsilon_0$. Initialize with $K=1000$ and
$\epsilon_0=0.01$. Iterate until $\epsilon<\epsilon_0$. Then
multiply $K$ by $1.26$, divide $\epsilon$ by $1.26$ and iterate with
the new constraints. The process of iterating and multiplying and
dividing by $1.26$ is repeated ten times, so that at the end of the
process $K \approx 10000$ and $\epsilon_0 \approx 0.001$.

Figure~\ref{fig: simulation results} shows simulation results
comparing the approximate iterative LP decoder with and without
merging of check nodes, as suggested in Section~\ref{sec:cuts
for_improved_LPD}. In the simulation we picked codes at random from
Gallager's $(3,6)$-regular ensemble of length $N=1002$. These codes
were transmitted over a binary symmetric channel. For the merging
process, check nodes are picked as follows. First, a list of short
cycles (up to maximum length of $6$) in the graph is generated using
the procedure proposed in~\cite{tiernan1970efficient}. Based on this
list, in each cycle we merge the check nodes into one new check
node. This approach is motivated by
Proposition~\ref{prop:2_equals_1}, which guarantees an unchanged
relaxation if the set of merged constraint nodes is cycle-free. As a
reference, we also plot results for iterative belief propagation
decoding. It is apparent from the figure that despite the
improvement due to check node merging, the LP decoder exhibits worse
performance than the belief propagation decoder for values of the
channel crossover probability greater than $0.045$, and better
performance for lower values of crossover probability.

In Figure~\ref{fig: mindistbarfracdist}, we plot our results from
Sections~\ref{sec:min distance bound} and~\ref{sec: lower bound on
dfrac}. The lower bound on the minimum distance derived in
Section~\ref{sec:min distance bound} is shown in circles. This plot
shows, for various block lengths, the average value of
$\underline{l_{\min}}$ (see \eqref{eq: lmin}) over $10$
randomly-generated codes, taken from Gallager's $(3,4)$-regular
ensemble. We also include improved bounds obtained by merging check
nodes. Check nodes were chosen for merging as in Figure~\ref{fig:
simulation results}, using all cycles of length up to $6$.

Figure~\ref{fig: mindistbarfracdist} also shows an improved lower
bound on the minimum distance using the bar method outlined at the
end of Section~\ref{sec:min distance bound}. Clearly, in this case
the lower bound is significantly improved by allowing pairs of nodes
to be examined. Next, we plot the lower bound on the fractional
distance from Section~\ref{sec: lower bound on dfrac}. The lower
bound on the fractional distance is calculated with the penalty
constant set to $B=10N$. In this case we set $\epsilon_0$ in
Algorithm~\ref{alg} to an initial value of $0.1/B$, instead of
$0.01$ as previously stated; this more stringent setting was used in
order to increase the accuracy of our result in light of the large
value of $\gamma_{\max}$\footnote{This is the setting which could
potentially have the effect on the convergence rate of
Algorithm~\ref{alg} which was discussed in the end of
Section~\ref{sec: lower bound on dfrac}.}. Finally, exact fractional
distance results, appearing in asterisks, are included in
Figure~\ref{fig: mindistbarfracdist}. These results were taken
from~\cite{lpdecode} and depict exact fractional distance, averaged
over a sample of $100$ codes. With the exception of these exact
results from~\cite{lpdecode}, all results depict the average over
the same $10$ randomly-generated codes, taken from Gallager's
$(3,4)$-regular ensemble. Comparing the exact results with our lower
bound, we see that although the randomly-selected codes are not the
same, the statistical averages indicate that the lower bound is
close to the mark for $B=10N$.

\section{Conclusion} \label{sec: conclusion}
In this paper we made obtained the following contributions to the
framework of LP decoding. First, a method for improving LP decoding
based on merging check nodes was proposed. Second, an algorithm for
determining a lower bound on the minimum distance was presented.
This algorithm can be improved by the check node merging technique
or by the greedy procedure introduced in Section~\ref{sec:min
distance bound}. This algorithm has computation complexity $O(N^2)$,
where $N$ is the block length. Third, an algorithm for determining a
tight lower bound on the minimum distance was presented. This
algorithm also has complexity $O(N^2)$. Fourth, we showed how the
fundamental polytope can be obtained for GLDPC and nonbinary codes.

\rem{
\section*{Acknowledgment}
The author would like to thank...
}

\bibliography{bibliography}

\end{document}